\documentclass[aps,floats,showpacs,amstex,amssymb,tightenlines,notitlepage,hidelinks,floatfix]{revtex4-1}

\usepackage{amsmath}
\usepackage{amsthm}
\usepackage{amsfonts}
\usepackage{amscd}
\usepackage{epsfig}
\usepackage{amssymb}
\usepackage{tabularx}
\usepackage{longtable}
\usepackage{calligra}
\usepackage{enumerate}
\usepackage{hyperref}
\usepackage{mathrsfs}
\def\j2{\mathbf{J}^2}
\def\a{\alpha}
\def\b{\beta}
\def\g{\gamma}

\def\e{\epsilon}
\def\ve{\varepsilon}
\def\f{\mathcal{F}}
\def\m{\mu}
\def\n{\nu}
\def\d{\delta}

\def\L{\Lambda}

\def\p{\partial}

\def\q{\mathcal{Q}}

\def\td{\widetilde{D}}
\def\hd{\widehat{D}}

\def\fn{\left. \Phi_n \right|_t }
\def\dfn{\left. \dot \Phi_n \right|_t }
\def\rn{Reissner-Nordstr\"om }

\newtheorem{thm}{Theorem}

\newtheorem{cor}[thm]{Corollary}

\begin{document}
\title{Black hole nonmodal linear stability: odd perturbations of  Reissner-Nordstr\"om}

\author{Juli\'an M.  Fern\'andez T\'{\i}o and Gustavo Dotti}
\affiliation {Facultad de Matem\'atica, Astronom\'{\i}a y
F\'{\i}sica (FaMAF), Universidad Nacional de C\'ordoba and\\
Instituto de F\'{\i}sica Enrique Gaviola, CONICET.\\ Ciudad
Universitaria, (5000) C\'ordoba,\ Argentina}

\begin{abstract}
Following a program on black hole nonmodal linear stability initiated in Phys.\ Rev.\ Lett.\  {\bf 112} (2014) 191101, 
we study odd linear perturbations of the Einstein-Maxwell equations around a Reissner-Nordstr\"om (A)dS 
black hole. We show that all the gauge invariant information in the metric and Maxwell field perturbations 
is encoded in the  spacetime scalars $\mathcal{F} =\d (F^*_{\alpha \beta} F^{\alpha \beta})$
 and $\mathcal{Q} =\delta (\tfrac{1}{48} C^*_{\alpha \beta \gamma \delta} C^{\alpha \beta \gamma \delta})$,  
where $C_{\alpha \beta \gamma \delta}$
 is the Weyl tensor, $F_{\alpha \beta}$ the Maxwell field, a star denotes Hodge dual and $\delta$ means 
first order variation, and that   the linearized Einstein-Maxwell equations are equivalent to a coupled 
system of wave equations  for  $\mathcal{F}$ and $\mathcal{Q}$. 
 For nonnegative cosmological constant 
we prove that $\mathcal{F}$ and $\mathcal{Q}$ are pointwise bounded  on the outer static region. 
The fields are shown to diverge as the Cauchy horizon is approached from the inner dynamical region, providing evidence 
supporting strong cosmic censorship. 
In the asymptotically AdS case
the dynamics depends on the boundary condition at the conformal timelike boundary and there are 
instabilities if Robin boundary conditions are chosen.
 \end{abstract}


\maketitle

\tableofcontents

\section{Introduction}

General Relativity coupled to Maxwell fields admits static charged black hole solutions    in spacetime dimensions $d=4$ and higher.  
The spacetimes  
 are  warped products ${\cal M} = {\cal N} \times_{r^2} \sigma^n$ 
of a two dimensional Lorentzian ``orbit manifold" ${\cal N}$ with  line element 
$\tilde g_{ab}(y) dy^a dy^b$
and  an $n=d-2$ dimensional Riemannian ``horizon manifold''  $\sigma^n$ with metric  $\hat g_{AB}(x)\,  dx^A dx^B$ (see, e.g., \cite{Kodama:2003kk}):
\begin{equation} \label{warped} 
g_{\a \b} dz^{\a} dz^{\b} = \tilde g_{ab}(y) \, dy^a dy^b + r^2(y) \hat g_{AB}(x)\,
  dx^A dx^B.
\end{equation}
In  four dimensions, the solution with $S^2$ horizon   is the Reissner-Nordstr\"om  black hole. If we use 
  the standard angular coordinates $\hat g_{ij}(x)\,  dx^i dx^j = d \theta^2 + \sin ^2 \theta \; d \phi^2$ and 
 static coordinates $(t,r)$ for the orbit manifold,  the Reissner-Nordstr\"om  metric is given by
\begin{equation} \label{rn}
  g_{\a \b} dz^{\a} dz^{\b}  = -f dt^2 + \frac{dr^2}{f} + r^2 (d \theta^2 + \sin ^2 \theta \; d \phi^2 ), 
  \end{equation}
where the norm $f$ 
   of the Killing vector $k^a = \p / \p t$ in (\ref{rn}) is 
\begin{equation} \label{f}
f = 1 - \frac{2M}{r} + \frac{Q^2}{r^2} - \frac{\L}{3} r^2.
\end{equation}
In (\ref{f}), $\Lambda$ is the cosmological constant, and $M$ and $Q$ are constants of integration 
that correspond to  mass and  charge respectively. The metric (\ref{rn}) together with the Maxwell 
field
\begin{equation} \label{max}
F = \frac{Q}{r^2} \; dt \wedge dr,
\end{equation}  
solves the  Einstein-Maxwell field equations
\begin{align} \label{efe}
&G_{\a \b} + \L g_{\a \b} = 8 \pi T_{\a \b}, \\  \label{max1}
& \nabla_{[\a} F_{\b \gamma]} =0,\\ \label{max2}
&\nabla^{\b} F_{\a \b} = 0,
\end{align}
where
\begin{equation} \label{efe2}
T_{\a \b} = \tfrac{1}{4\pi} \left( F_{\a \g} F_{\b}{}^{\g} -\tfrac{1}{4} g_{\a \b} F_{\g \d}F^{\g\d} \right).
\end{equation}
Note that, since $T_{\a \b}$ is traceless, (\ref{efe}) is equivalent to 
\begin{equation} \label{efe1}
R_{\a \b} - \L g_{\a \b} = 8 \pi T_{\a \b}
\end{equation}
\noindent
We assume $Q \neq 0$ and focus in the black hole cases, which are those 
 for which  there is an outer  static ($f>0$) region, that is, 
either $\Lambda \leq 0$ and   $0<r_h<r<\infty$, or 
 $\Lambda > 0$ and  
$0<r_h<r<r_c$. Here the  event and cosmological  horizons $r=r_h$ and $r=r_c$ 
are simple zeros of the quartic polynomial $r^2 f$ if the black hole is non extremal.
 The range of values of $M, Q$ and $\Lambda$ 
giving black holes  can be found 
in Appendix A of  \cite{Kodama:2003kk}. \\

We are interested in proving the non-modal linear stability of the outer static region of 
the solution (\ref{rn})-(\ref{max}) of the field equations (\ref{efe})-(\ref{efe2}). By this we mean  \cite{Dotti:2013uxa}
\cite{Dotti:2016cqy}, 
showing that:
\begin{enumerate}[i)]
\item there are  gauge invariant (both in the Maxwell and linear gravity sense) 
scalar fields $\chi: {\cal M} \to \mathbb{R}$ that contain the same information as  the perturbation 
$\mathcal{F}_{\a \b} = \d F_{\a \b}$ of the electromagnetic 
  field and the gauge class $[h_{\a \b}]$ of the metric perturbation $h_{\a \b} =\delta g_{\a \b}$, 
 and measure the distortion of the geometry and the Maxwell field. 
By ``contain the same information'' we mean that  
 $h_{\a \b}$  in a given gauge  and   $\mathcal{F}_{\a \b}$  
can be obtained by applying some injective  linear functional on the fields $\chi$. 
\item the fields $\chi$ are pointwise bounded on the outer static region by constants that depends on 
the initial data of the perturbation on a Cauchy surface.
\end{enumerate}

The perturbed metric and Maxwell fields  can be expanded in  series involving 
 rank 0,1, and 2 eigentensor fields of the horizon manifold Laplace-Beltrami (LB) operator,  with ``coefficients'' that are  
tensor fields in the orbit space $\mathcal{N}$;  this is the mode expansion of $h_{\a \b}$ and 
$\mathcal{F}_{\a \b}$ \cite{Kodama:2003kk}.
 The linearized Einstein-Maxwell equations (LEME) do not mix modes. A master scalar field  
$\mathcal{N} \to \mathbb{R}$ can be extracted for each mode such that the LEME 
reduce to an infinite set of {\em scalar wave equations on} $\mathcal{N}$, one for each master scalar field. 
This was proved in four dimensional General Relativity 
in the seminal black hole stability papers \cite{Regge:1957td} \cite{Zerilli:1970se} \cite{Zerilli:1974ai}  and in higher dimensions more recently by Kodama and 
Ishibashi  (see, eg.g. \cite{Kodama:2003kk} and \cite{Ishibashi:2011ws}). 
Prior notions of linear stability are based on   the boundedness of the master fields  on the orbit manifold $\mathcal{N}$, we call
this {\em modal (linear) stability}. 
In the case of  four dimensional asymptotically flat  charged black holes  the modal linear stability 
  was  proved by Zerilli and  Moncrief  in the series 
of articles \cite{Zerilli:1974ai, Moncrief:1974gw, Moncrief:1974ng, Moncrief:1975sb} (see also \cite{wald}) \\

The limitations of the modal linear stability notion are explained   in \cite{Dotti:2013uxa} and 
\cite{Dotti:2016cqy}, where a non modal stability  concept based on i) and ii) above was proved to hold 
for the Schwarzschild and Schwarzschild de Sitter black holes. In these papers the fields $\chi$ in i) 
 are 
 gauge invariant combinations of perturbations of scalars made out of contractions of
the Weyl tensor, its dual, and its first covariant derivative. \\

For Einstein-Maxwell black holes the extra degrees of 
freedom coming from the Maxwell field have to be accounted for.  Perturbations 
naturally split into two decoupled types: odd and even (Section \ref{rwz}). In this paper  we  prove  
the non modal linear stability of the Reissner-Nordstr\"om  black hole under odd perturbations. The fields  $\chi$ 
that fulfill the requirements i) and ii) above 
 are the 
 first order perturbation of the scalars obtained  by contracting the Maxwell and Weyl tensors with their 
Hodge duals: 
$\mathcal{F} =\d (F^*_{\alpha \beta} F^{\alpha \beta})$
 and $\q =\delta (\tfrac{1}{48} C^*_{\alpha \beta \gamma \delta} C^{\alpha \beta \gamma \delta})$. 
These fields are shown to satisfy a coupled system of wave equations in the Reissner-Nordstr\"om  background, and this 
fact is used to prove their pointwise boundedness on the outer static region. 
We defer to future work the treatment of even perturbations.

\section{Linearized Einstein-Maxwell equations} \label{rwz}

Let   $(g({\ve})_{\a \b}, F(\ve)_{\a \b})$ be a  one-parameter family 
of solutions of the Einstein-Maxwell equations (\ref{efe2})-(\ref{efe1}),  with $g({\ve=0})_{\a \b}$ and 
$F(\ve=0)_{\a \b}$ the 
 Reissner-Nordstr\"om  fields  (\ref{rn})-(\ref{max}). Note that  all fields in this paper are assumed to be  jointly  smooth
in the spacetime coordinates and (in the case of one-parameter families) the perturbation parameter. The  {\em perturbation fields}
\begin{equation} \label{lf}
h_{\a \b} \equiv \frac{d}{d \ve}\bigg|_{\ve=0}g({\ve})_{\a \b}, \;\;\; \f_{\a \b} \equiv \frac{d}{d \ve}
\bigg|_{\ve=0}F({\ve})_{\a \b}
\end{equation}
satisfy the linearized Einstein-Maxwell equations (LEME):
\begin{align} \label{1}
&\frac{d}{d \ve}\bigg|_{\ve=0} G_{\a \b}(g(\ve)) +  \L h_{\a \b} = 2 \frac{d}{d \ve}\bigg|_{\ve=0}  \left[ 
F({\ve})_{\a \g} F({\ve})_{\b \m} g(\ve)^{\g \m} - \tfrac{1}{4} g({\ve})_{\a \b} \left( F({\ve})_{\m \n} F({\ve})_{\m' \n'} 
g(\ve)^{\m \m'}g(\ve)^{\n \n'} \right) \right],\\ \label{2}
&\frac{d}{d \ve}\bigg|_{\ve=0} \p_{[\a} F({\ve})_{\b \g]} =0,\\ \label{3}
&\frac{d}{d \ve}\bigg|_{\ve=0} \nabla_{\b} F({\ve})^{\a \b} =0.
\end{align}
As in equation (\ref{warped}), we adopt  the notation in \cite{Chaverra:2012bh} and 
  use lower case indexes $a, b, c, d, e$ for tensors on the orbit manifold ${\cal N}$, 
upper case indexes   $A, B, C, D,...$ 
for tensors on $S^2$, and Greek indexes  for space-time tensors,  and follow the additional  convention 
 in \cite{Dotti:2016cqy} that 
\begin{equation} \label{index}
\a=(a,A), \b=(b,B), \g=(c,C), \d=(d,D),... \text{ etc. }
\end{equation}
Tensor fields {\em   introduced}
 with a lower $S^2$ index (say $Z_A$) and {\em then shown with an upper $S^2$ index} are assumed 
to have been acted  on {\em with the unit $S^2$ metric inverse $\hat g^{AB}$},
 (i.e., in our example, $Z^A \equiv \hat g^{AB} Z_B$), 
and similarly with upper $S^2$ indexes moving down.  This has to be kept in mind to avoid wrong $r^{\pm2}$ factors 
in the equations. 
$\td_a, \tilde \e_{ab}$ and $\tilde g^{a b}$ are the covariant derivative, volume form (any chosen orientation) and metric inverse 
for the $\mathcal{N}$ orbit space; $\hd_A$ and  $\hat \e_{AB}$ are the covariant derivative and volume form $\sin (\theta) d\theta \wedge d\phi$ 
on the unit sphere. As an example, the Laplacian on 
scalar fields can be written in terms of the differential operators $\td_a$ and $\hd_A$ as 
\begin{equation} \label{lap}
\nabla_{\a} \nabla^{\a} \Phi = \td_a \td^a \Phi + \frac{2}{r} (\td^b r) (\td_b \Phi) + \frac{1}{r^2} \hd_A \hd^A \Phi.
\end{equation}
The linearized field equations,  (\ref{lf}) and (\ref{2}) imply that locally there exists a vector potential $A_{\a}$ such that 
\begin{equation}
\mathcal{F}_{\a \b} = \p_{\a} A_{\b} - \p_{\b} A_{\a}, 
\end{equation}
The linear fields entering the LEME are  $h_{\a \b}$ and $A_{\a}$.
 Under the index 
convention (\ref{index}) the covector field $A_{\a}$ is written
\begin{equation} \label{vd}
A_{\a}=(A_a, A_A). 
\end{equation}
From the $S^2$ viewpoint $A_{\a}$  contains two scalar fields $A_a^+ \equiv A_a$ and a covector field $A_A$. 
Using Proposition 2.1 in \cite{Ishibashi:2004wx} and the fact that the first Betti number of $S^2$ is 
zero (which implies that divergence free $S^2$ covectors are of the form $\hat \e_{A}{}^B \hd_B P$, with 
$P$ 
an $S^2$ scalar field), 
we can write  $A_A = \hd_A A^+ + \hat \e_{A}{}^C \hd_C A^-$, thus
\begin{equation} \label{Adecomp}
A_{\a}=(A^+_a, \hd_A A^+ + \hat \e_{A}{}^C \hd_C A^-).
\end{equation}
The scalar fields $A^{\pm}$ are unique if they are required to belong to $L^2(S^2)_{>0}$ \cite{Dotti:2016cqy}, 
where $L^2(S^2)_{>\ell_o}$
is the  space 
of square integrable functions on $S^2$ orthogonal to the $\ell=0, 1,...,\ell_o$ eigenspaces of the LB operator, and 
 $\ell$ labels the LB scalar field eigenvalue $-\ell (\ell+1)$. The plus (even) and minus (odd) signs on tensor fields refer to the way 
they transform when pull backed by the 
antipodal map $P$ on $S^2$ \cite{Dotti:2013uxa}. \\
A symmetric tensor field
 $S_{\a \b} = S_{\b \a}$,
\begin{equation} \label{std1}
S_{\a \b} = \left( \begin{array}{cc} S_{ab} & S_{a B} \\ S_{A b} & S_{AB} \end{array} \right),
\end{equation}   
 such as the perturbations of the metric,  the Einstein and  the energy momentum tensor fields, 
contains  three $S^2$ scalar fields $S_{ab}^+ \equiv S_{ab}$, 
two $S^2$ covector fields $S_{a B}$  and a  
symmetric tensor field $S_{AB}$. The $S^2$ covectors can be decomposed as in (\ref{Adecomp})
\begin{equation} \label{std2}
S_{a B} = \hd_B S_a^+ + \hat \e_{B}{}^C \hd_C S_a^-,
\end{equation}
where again $S_a^{\pm}$ are unique if their components are in $L^2(S^2)_{>0}$. 
From  Proposition 2.2 in  \cite{Ishibashi:2004wx} and the fact that there are  no 
  transverse traceless symmetric rank two  tensor fields on $S^2$, follows that  
\begin{equation} \label{std3}
S_{AB} = \hd_{(A} ( \e_{B) C} \hd^C S^-) + \left( \hd_A \hd_B - \tfrac{1}{2} \hat g_{AB} \hd^C \hd_C \right) S^+ + \tfrac{1}{2} \; S_T^+ \; \hat g_{AB} , \;\; (S_T^+ = S_C{}^C).
\end{equation}
The fields $S^{\pm}$ are unique if required to belong to $L^2(S^2)_{>1}$ \cite{Dotti:2016cqy}. 
In this way, the symmetric field $S_{\a \b}$ 
is replaced by two sets of fields, even (+)  and odd (-):
\begin{equation} \label{std4}
S_{\a \b} \sim \{ S_{ab}^+=S_{ab}, \; S_a^+, \; S^+, \; S_T^+ \} \cup \{ S_a^-, \; S^- \}.
\end{equation}
If we decompose the linearized symmetric tensor fields $h_{\a \b}$, $d G_{\a \b} / d\ve |_0 \equiv \mathcal{G}_{\a \b}$  
and $d T_{\a \b}/ d\ve |_0 \equiv \mathcal{T}_{\a \b}$ as in 
(\ref{std1})-(\ref{std4}), we get  the following sets  of even and odd fields: 
\begin{align} \label{hdecomp}
{h}_{\a \b} &\sim \{  h_{ab}^+, \; h_a^+, \; h^+, \; {h}_T^+ \} \cup \{ h_a^-, \; h^- \}\\
\label{Gdecomp}
\mathcal{G}_{\a \b} &\sim \{  G_{ab}^+, \; G_a^+, \; G^+, \; G_T^+ \} \cup \{ G_a^-, \; G^- \}\\
\mathcal{T}_{\a \b} &\sim \{  T_{ab}^+, \; T_a^+, \; T^+, \; T_T^+ \} \cup \{ T_a^-, \; T^- \}. \label{Tdecomp}
\end{align}
Group theoretical arguments   (refer to  Section 2 of \cite{Ishibashi:2004wx})  indicate that 
the LEME  involving the even fields 
 in (\ref{Adecomp}) (\ref{hdecomp})-(\ref{Tdecomp})  decouple from those 
involving  the 
odd  fields, so we can switch off one sector and study purely odd or even perturbations. 
Odd perturbations are the subject of this paper. \\

 We will find it useful to introduce 
the square angular momentum operator
\begin{equation} \label{j2}
\j2 \equiv (\pounds_{J_{(1)}})^2 + (\pounds_{J_{(2)}})^2 + (\pounds_{J_{(3)}})^2, 
\end{equation}
where $J_{(1)}, J_{(2)}$ and $J_{(3)}$ are $S^2$ (and thus spacetime) Killing vector fields corresponding to rotations 
around orthogonal axis in $\mathbb{R}^3 \supset S^2$, with maximum orbit orbit lengths set to $2 \pi$ 
(e.g., $J_{(3)} = \p / \p_{\phi}$). On $S^2$ scalar fields the operator $\mathbf{J}^2$ agrees with the $S^2$ LB 
 operator 
$\hd^A \hd_A$, but these two operators 
act differently on higher rank tensors. A key property of $\mathbf{J}^2$ is that
  it commutes 
with $\nabla_{\a}$, $\td_a$ and $\hd_A$, this follows from 
 $[\nabla_a, \pounds_{J_k}] = 0 = [\hd_A, \pounds_{J_k}] = [\td_a, \pounds_{J_k}]$.
The modal decomposition consists in expanding the $S^2$ scalars in  (\ref{Adecomp}) (\ref{hdecomp})-(\ref{Tdecomp}) 
in a real basis of spherical harmonics of $S^2$, which are 
eigenfields of $\mathbf{J}^2$  with eigenvalues $-\ell(\ell+1)$, the eigenspaces being of of dimension $2\ell +1$. 
The differential operators that give a symmetric tensor $S_{\a \b}$ or a covector $A_{\a}$ in terms of these $S^2$ scalars 
commute with $\mathbf{J}^2$. Thus,  if the $S^2$ scalar 
fields in (\ref{hdecomp}) and (\ref{Adecomp}) lie  on the $\ell$ eigenspace, 
then $h_{\a \b}, \mathcal{G}_{\a \b}, \mathcal{F}_{\a \b}$ will all be eigentensors of $\mathbf{J}^2$ with eigenvalue 
$-\ell(\ell+1)$, i.e., 
different modes stay unmixed. The distinction between even and odd modes can now be stated in a precise way: 
if $X^{\pm}$ is  a  covector  (\ref{Adecomp}) or symmetric field $S_{\a \b}$ (\ref{std1})-(\ref{std4}) of 
a given parity,  made out of  scalars of 
harmonic numbers ($\ell,m)$, then 
$\j2 X^{\pm} = -\ell(\ell+1) X^{\pm}$ and  $P_* X^{\pm} = \pm (-1)^{\ell} X^{\pm}$. \\

We will assume that  $A^{\pm}, S_a^{\pm}$ to $L^2(S^2)_{>0}$ and $S^{\pm}$ to $L^2(S^2)_{>1}$, 
since  then the  linear operators $(A^+, A^-) \to A_{\a}$ in (\ref{Adecomp}), 
and $\{ S_{ab}^+, \; S_a^+, \; S^+, s^+ \} \cup \{ S_a^-, \; S^- \} \to S_{\a \b}$ 
in (\ref{std3}) are injective \cite{Dotti:2016cqy}.
Consequently, the odd sector LEME (\ref{1})  are equivalent to
\begin{align} \label{Ea}
&G^{-}_a + \Lambda h^{-}_a = 8\pi T_a^{-},\\ \label{E}
&G^{-} + \Lambda h^{-} = 8\pi T^{-}.
\end{align}

\subsection{Odd sector perturbations} \label{osp}

Odd perturbations are those for which the plus fields in (\ref{Adecomp})  and (\ref{hdecomp}) are  zero, that is 
\begin{equation} \label{pert1} \renewcommand*{\arraystretch}{1.3}
h_{\a \b} = \left( \begin{array}{cc} 0 & \hat \e_{B}{}^C \hd_C h_a^- \\ \hat \e_{A}{}^C \hd_C h_b^- & 
\hd_{(A} ( \e_{B) C} \hd^C h^-)  \end{array} \right), \;\;\; \;\;\;
\f_{\a \b} = \left( \begin{array}{cc} 0 &  \td_a  (\hat \e_{B}{}^C \hd_C A^-) \\ -\td_b ( \hat \e_{A}{}^C   \hd_C A^-)& 
- \e_{AB} \hd^C\hd_C A^- \end{array} \right),
\end{equation}  
with $A^- , h_a^- \in L^2(S^2)_{>0}$ and  $h^- \in L^2(S^2)_{>1}$, which are conditions that guarantee their uniqueness, 
as explained at the end of the previous Section.
$U(1)$ gauge transformations of the Maxwell field are of the form 
$A_{\a} \to A_{\a } + \p_{\a} B$ 
and therefore affect only the even piece of the vector potential (\ref{Adecomp}) leaving $A^-$ invariant. \\
Under a coordinate gauge transformation (infinitesimal diffeomorphism)  along the odd vector field $\zeta^{\a}=(0, \hat \e^{AB} \hd_B \xi)$, 
$\xi \in L^2(S^2)_{>0}$, $h_{\a \b}$ and $\f_{\a \b}$ transform into the  physically equivalent  fields:
\begin{equation} \label{prime}
h'_{\a \b} = h_{\a \b} + \pounds_{\zeta} g_{\a \b}, \;\; \f'_{\a \b}=\f_{\a \b} + \pounds_{\zeta} F_{\a \b}=\f_{\a \b}.
\end{equation}
We call $\mathcal{L}_-$ the set of odd solutions $(h_{\a \b},\f_{\a \b})$ of the LEME (\ref{1})-(\ref{3}) mod 
the equivalence relation $h_{\a \b} \sim h'_{\a \b}$ above, that is, if $[h_{\a \b}]$ denotes equivalence class 
under  the first  transformation (\ref{prime}), then 
\begin{equation} \label{L-}
\mathcal{L}_- = \{ ([h_{\a \b}], \f_{\a \b}) \; | \; (h_{\a \b}, \f_{\a \b}) \; \text{ is an odd  solution of  (\ref{1})-(\ref{3}) } \}.
\end{equation}
The transformation  (\ref{prime}) 
 is equivalent to 
\begin{equation} \label{gauge}
h_a^- \to h_a^- +r^2 \td_a \xi, \;\;\; h^- \to h^-+ r^2 \xi_{>1}, \;\;\; A^- \to A^-, 
\end{equation}
and implies  that the field $A \equiv A^-$ is gauge invariant. If we project  
  $h_a^-  = (h_a^-)^{(\ell=1)}+ (h_a^-)^{(>1)}$ 
 onto its   $L^2(S^2)_{(\ell=1)}$ and $L^2(S^2)_{>1}$ pieces, and similarly for 
the other fields, and keep in mind that $h^- = (h^-)^{(>1)}$,  we find that:  
\begin{enumerate}[(i)]
\item  The  $\mathcal{N}$ 1-form $h_a^{>1}  \in L^2(S^2)_{>1}$ defined by 
\begin{equation} \label{gauge1}
h_a^{>1} \equiv (h_a^-)^{>1} -r^2 \td_a (r^{-2} h^-)
\end{equation}
 is gauge invariant. 
\item  There exists a gauge  for $(h_{\a \b})^{>1}$ such that  $h^-=0$. 		In view of (\ref{gauge1}), 
in this gauge
\begin{equation} \label{pertrw} \renewcommand*{\arraystretch}{1.3}
(h_{\a \b})^{>1} = \left( \begin{array}{cc} 0 & \hat \e_{B}{}^C \hd_C h_a^{>1} \\ \hat \e_{A}{}^C \hd_C h_b^{>1} & 
0  \end{array} \right)
\end{equation}  
 This is the well known  Regge-Wheeler (RW) gauge for $(h_{\a \b})^{>1}$, and is unique,
in the sense that, according to  (\ref{gauge}), 
applying to (\ref{pertrw})  any gauge transformation that is non trivial 
in the $\ell>1$ sector, spoils the $h^-=0$ condition.
\item For $\ell=1$ the only possible gauge invariant metric field is \cite{Sarbach:2001qq} 
\begin{equation} \label{z}
\mathcal{Z} :=\tilde \e^{cd} \td_c \left(  \frac{h_d^{(\ell=1)}}{r^2} \right)
\end{equation}
\end{enumerate}
From now on we  work in  RW gauge (\ref{pertrw}),
 then   we set $h^-=0$  in (\ref{pert1}) and replace $h_a^-$ with $h_a=(h_a)^{(\ell=1)} 
+ h_a^{>1}$. 
With this choice 
 the absolute value $g$ of the determinant of the metric  agrees (to linear order) with the absolute value $g_o$ 
of the unperturbed metric determinant, then 
\begin{equation} \label{det}
\sqrt{g} =  \sqrt{g_o} = r^2 (\tilde g)^{1/2} (\hat g) ^{1/2}, \;\; \tilde g = - \text{det}(\tilde g_{ab}),  \;\; \hat g =  \text{det}(\hat g_{AB}).
\end{equation}
To linear order the  inverse metric is
$$
g^{\a \b} = \left( \begin{array}{cc} \tilde g^{ab} & -\ve \; r^{-2} \, \hat \e^{BC} \hd_C h^a \\ -\ve \; r^{-2}\, \hat \e^{AC} \hd_C h^b
 & r^{-2} \; \hat g^{AB} \end{array} \right). 
$$
This is  used   to raise the indexes of the perturbed Maxwell field $F_{\a \b} =  Q r^{-2} \tilde \e_{ab} + \ve \f_{\a \b}$. The result is 
\begin{align}\nonumber
F^{ab} &= E_o \tilde \e^{ab}, \;\;\; E_o \equiv Q r^{-2}\\
F^{aB} &= \ve \; r^{-2} \hat \e^{BC} \hd_C \td^a A - \ve \; E_o r^{-2} \tilde \e^{ad} \hat \e^{BC} \hd_C h_d \;\;  (= -F^{Ba})\\
F^{AB} &= - \ve \,r^{-4} \hat \e^{AB} \hd^C \hd_C A \nonumber
\end{align}
Taking advantage of (\ref{det})  we find that 
 the linearized  Maxwell  equation (\ref{3}) can be 
written 
\begin{equation}
0 =\frac{1}{\sqrt{g}} \p_{\a} \left(\sqrt{g} F^{\a \b} \right) = \frac{1}{\sqrt{g_o}} \p_{\a} \left(\sqrt{g_o} (F_o^{\a \b} + \ve \f^{\a \b})\right) =
\ve \frac{1}{\sqrt{g_o}} \p_{\a} \left(\sqrt{g_o} \, \f^{\a \b} \right). 
\end{equation}
and using  $\sqrt{g_o}= r^2  (\tilde g)^{1/2} (\hat g) ^{1/2}$,  we find that the $\beta=b$ equation above is trivial 
whereas the $\beta=B$ equation gives  
\begin{equation}\label{divf1}
0 =  \frac{1}{\sqrt{g_o}} \p_{\a} \left(\sqrt{g_o} \, \f^{\a B} \right) =  r^{-2} \hat \e^{BC} \hd_C\; 
 [ \td^a \td_a A  + r^{-2} \hd^D \hd_D A-
 \tilde \e^{ad} \td_a \left( Q r^{-2} h_d \right)], 
\end{equation} 
which, since $A, h_a \in L^2(S^2)_{>0}$,  is equivalent to 
\begin{equation}\label{divf}
0 =  \td^a \td_a A  + r^{-2} \hd^D \hd_D A-
 \tilde \e^{ad} \td_a \left( Q r^{-2} h_d \right).
\end{equation}

The linearized Einstein's equations (\ref{1}) are equivalent to the set (\ref{Ea}) and (\ref{E}). After a lengthly calculation we find 
\begin{equation} \label{G-T-}
G^- = \td^a h_a^{>1}, \;\;T^- = 0.
\end{equation}
 We also find that 
\begin{equation} \label{Ta}
8 \pi T_a^- = -\frac{2Q}{r^{2}} \tilde  \e_a{}^b \td_b A + \frac{Q^2}{r^{4}} h_a
\end{equation}
and 
\begin{equation} \label{Ga}
-2r^2 G_a^- = \tilde \e_a{}^b \td_b \left( r^4 \, \tilde \e^{cd} \td_c \left(\frac{h_d}{r^2} \right) \right) +\hd^B \hd_B h_a +
\left( \td^c \td_c r^2 + 4 r^2 \Lambda \right) h_a.
\end{equation}

\subsubsection{$\ell>1$ modes.}

Since $h^-=0$, equations (\ref{E}) and (\ref{G-T-}) give $\td^a h_a^{>1}=0$. The solution of this equation is 
\begin{equation}\label{ha}
h_a^{>1} = \tilde \e_a{}^b \td_b (Z),  \;\; \; Z \in L^2(S^2)_{>1}
\end{equation}
for some potential $Z$,  defined  up to the sum of a  function of $(\theta, \phi)$:
\begin{equation} \label{Zf}
Z(t,r,\theta,\phi) \to Z(t,r,\theta,\phi)+q(\theta,\phi).
\end{equation}
Inserting (\ref{Ta}), (\ref{Ga})  and (\ref{ha})  into the projection onto $L^2(S^2)_{>1}$ of the  linearized Einstein equation (\ref{Ea}) gives 
\begin{equation}\label{above}
\tilde \e_a{}^b \td_b \left[ r^4 \td^c \left(\frac{\td_c Z}{r^2} \right) + \hd^B \hd_B Z \right] + 
\left( \td^c \td_c r^2 + 2 r^2 \Lambda + 2 \frac{Q^2}{r^2} \right)  \e_a{}^b \td_b Z -4Q \e_a{}^b \td_b A^{>1} = 0
\end{equation}
The fact that 
\begin{equation} \label{fact}
\td^c \td_c r^2 + 2 r^2 \Lambda + 2 \frac{Q^2}{r^2} =2, 
\end{equation}
makes it possible to pull the operator $\e_a{}^b \td_b$ to the left in (\ref{above}) 
Since the kernel of $\e_a{}^b \td_b$ acting on $\mathcal{N}$-scalar fields are the 
$\mathcal{N}$-constants (i.e., functions of $(\theta,\phi)$), 
we can   lift $\e_a{}^b \td_b$ from this equation and get
\begin{equation} \label{eq1d2}
r^4 \td^c \left(\frac{\td_c Z}{r^2} \right) + \hd^A \hd_A Z + 2 Z =4QA^{>1} + z(\theta,\phi).
\end{equation}
We now use the freedom (\ref{Zf}) and choose $q(\theta,\phi)$ to cancel $z(\theta,\phi)$. This is possible 
since the operator $Z \to \hd^B \hd_B Z + 2 Z$ is invertible in $L^2(S^2)_{>1}$. This choice of $Z$ is equivalent to setting 
$z(\theta,\phi)=0$ in (\ref{eq1d2}). The resulting  equation is 
 equivalent to  the four dimensional wave equation
\begin{equation} \label{master}
\nabla_{\a} \nabla^{\a}  \Phi  + \left( \frac{8M}{r^3}- \frac{6Q^2}{r^4}- \frac{2 \L}{3} \right) \Phi
= \frac{4Q}{r^3} \; W^{>1},
\end{equation}
where
\begin{equation} \label{wfi0}
W^{>1}= \frac{A^{>1}}{r}, \;\;\; \Phi= \frac{Z}{r^2} \;\;  \in L^2(S^2)_{>1}.
\end{equation}

The equation obtained  after inserting (\ref{ha}) into the projection onto $L^2(S^2)_{>1}$ of the linearized Maxwell equation (\ref{divf}) 
and then using (\ref{eq1d2}), 
\begin{equation} \label{mas2}
\td^a \td_a A^{>1} + \frac{ \hd^B \hd_B A^{>1}}{r^2}-\frac{4Q^2}{r^4} A^{>1} = -\frac{Q^2}{r^4} (\j2+2) Z,
\end{equation}
also admits the form of a four dimensional wave equation linking $W$ and $\Phi$ above:
\begin{equation} \label{master2}
\nabla_{\a} \nabla^{\a} W^{>1} +  \left( \frac{2M}{r^3}- \frac{6Q^2}{r^4}- \frac{2 \L}{3} \right) W^{>1}
= - \frac{Q}{r^3} (\j2+2)\Phi.
\end{equation}
Here we used the facts that on scalar fields $\hd^A \hd_A=\j2$ and  $\td^c \td_c r = df/dr$. \\

Note that all steps above can be reversed: the system of equations (\ref{master}) (\ref{master2}) is equivalent 
to the system (\ref{eq1d2}) (\ref{mas2}) which, using (\ref{ha}) and the definitions (\ref{wfi0}) imply the LEME.
We conclude that the odd sector $\ell>1$ LEME are entirely equivalent to the system of (four dimensional) wave equations 
(\ref{master}) and 
(\ref{master2}) coupling the fields $\Phi$ and $W$. These fields 
 are potentials from which the ($\ell>1$ piece of the) metric perturbation in the RW gauge ${}^{RW}h_{\a \b}$
is given by equations (\ref{pertrw}), (\ref{ha}) and (\ref{wfi0}), 
and that of the electromagnetic field perturbations  by the second equation (\ref{pert1}) with $A^-=A^{>1}=rW$.  
The map $(\Phi,W) \to ({}^{RW}h_{\a \b},F_{\a \b})$ is injective. Otherwise, there is a $(\Phi_o,W_o)\neq (0,0)$ sent to $(0,0)$.  
In view of the second equation (\ref{pert1}) and $r W_o = A^{>1}$, it must be $\hd^C \hd_C W_o=0$ and therefore $W_o=0$ which, 
inserted in (\ref{master2}), gives   $ (\j2+2)\Phi=0$, and this is equivalent to $\Phi=0$ since  $\Phi \in L^2(S^2)_{>1}$. We conclude
that $W_o=\Phi_o=0$. 

\subsubsection{$\ell=1$ modes.}

The projection of the linearized Maxwell equation (\ref{divf1}) onto the three dimensional $\ell=1$ subspace  $L^2(S^2)_{\ell=1} \subset L^2(S^2)$  is 
\begin{equation}\label{divf-l=1}
Q \mathcal{Z}=  \td^a \td_a A^{(\ell=1)}  -2  \frac{ A^{(\ell=1)}}{r^2},
\end{equation}
where $\mathcal{Z}$, introduced in (\ref{z}),  is the only gauge invariant field of the $\ell=1$ metric perturbation (see (\ref{gauge})). \\
The projection of the LEME  (\ref{Ea}), using (\ref{Ta}), (\ref{Ga})  and (\ref{fact}) is
\begin{equation} \label{leme-l=1}
\tilde \e^{a b} \td_b \left[ r^4 \mathcal{Z} - 4 Q A^{(\ell=1)} \right]=0, 
\end{equation}
this implies that $ r^4 \mathcal{Z} - 4 Q A^{(\ell=1)}$ is a function of $(\theta,\phi)$ that, for convenience, 
we call $6M S(\theta,\phi)$, therefore 
\begin{equation} \label{za}
\mathcal{Z} = \frac{4QA^{(\ell=1)} + 6M  S(\theta,\phi)}{r^4},
\end{equation}
Since  both $\mathcal{Z}$ and 
$A$ belong to $L^2(S^2)_{(\ell=1)}$, it must be 
\begin{equation} \label{S1}
S= \sqrt{ \frac{4 \pi}{3}} \sum_{m=1}^3 j^{(m)} S_{(\ell=1,m)}. 
\end{equation}
where the $S_{(\ell=1,m)}$  are a real orthonormal basis of $L^2(S^2)_{\ell=1}$, such as 
\begin{equation}
S_{(\ell=1,m=1)}=\sqrt{ \frac{3}{4 \pi}} \; \sin (\theta) \cos(\phi), \; 
S_{(\ell=1,m=2)}=  \sqrt{ \frac{3}{4 \pi}} \; \sin (\theta) \sin(\phi), \;
S_{(\ell=1,m=3)}= \sqrt{ \frac{3}{4 \pi}} \cos(\theta).
\end{equation}
Inserting (\ref{za}) in (\ref{divf-l=1})  gives 
\begin{equation} \label{eqa1}
   \td^a \td_a  A^{(\ell=1)}  -\left(  \frac{2}{r^2}+\frac{4Q^2}{r^4}\right) A^{(\ell=1)} = \frac{6MQ}{r^4} S.
\end{equation}
The general solution of  the $\ell=1$ equations is therefore obtained by choosing $S(\theta,\phi)$ (equivalently,
 the  $j^{(m)}$ in (\ref{S1}), which, 
 as we will show below, are 
 infinitesimal angular momentum components) and a solution $A^{(\ell=1)}$ of (\ref{eqa1}). Then   
  $\mathcal{Z}$ is given by (\ref{za}) and 
$h_a$ obtained, mod gauge transformation, from (\ref{z}). \\

A particular solution of the {\em inhomogeneous} equation (\ref{eqa1}) {\em when} $
S=a \cos (\theta) \propto S_{(\ell=1,3)}$  is obtained 
by  considering the  KN(A)dS black hole solution  with mass $M$ and angular momentum $J=aM$ along 
the $\theta=0$ axis in Boyer 
 Lindquist coordinates (see, e.g., \cite{Gibbons:1977mu}, equations (2.19)-(2.24)),  and letting the angular momentum play 
the role of $\varepsilon$ in (\ref{lf}).
If 
we Taylor expand the metric around $a=0$ we obtain
\begin{equation}
g_{\a \b} = g_{\a \b}^{RN} +  h_{\a \b} + \mathcal{O}(a^2)
\end{equation}
where $g_{\a \b}^{RN}$ is the Reissner-Nordstr\"om  metric (\ref{rn})-(\ref{f}), 
\begin{align} \label{hti}
h_{\phi t} &= h_{t \phi} = a (f-1) \sin^2 (\theta) = \hat \e_{\phi}{}^{\theta} \p_{\theta} h_t,\\
h_{\theta t} &= h_{t \theta} =  0 =  \hat \e_{\theta}{}^{\phi} \p_{\phi} h_t, \label{hfi}
\end{align}
the remaining components being trivial. We recognize that $h_{\a \b}$ is an $\ell=1$  perturbation 
with $j^{(1)}=j^{(2)}=0$. 
 Since $\widehat \e = \sin (\theta) \; d\theta \wedge d\phi$, 
equations  (\ref{hti})-(\ref{hfi})  and 
$0= h_{r \phi}= h_{r \phi}$  imply that 
\begin{equation} \label{ha-1}
h_t= a(f-1) \cos(\theta), \;\;\; h_r=0, 
\end{equation}
which, inserted in (\ref{z}) gives 
\begin{equation}\label{Z0}
\mathcal{Z} = a \cos(\theta) \; \frac{6Mr-4Q^2}{r^5} =: \mathcal{Z}_{KN}^o
\end{equation}
The nonzero components of 
the Maxwell vector potential $A_{\a}$ for the electromagnetic field $F_{\a \b}$ of the KN(A)dS black hole are   (equation (2.24)
 in \cite{Gibbons:1977mu})
\begin{equation} \label{A-1}
A_t= \frac{Q}{r} + \mathcal{O}(a^2), \;\; A_{\phi} = -\frac{Q \sin^2(\theta)}{r} \; a +  \mathcal{O}(a^3), \;\; A_r=A_{\theta}=0,
\end{equation}
whose exterior derivative, consistently, gives a $j^{(1)}=j^{(2)}=0$, $\ell=1$ odd perturbation  of the 
electromagnetic
 field with (see the second equation (\ref{pert1}))
\begin{equation} \label{A0}
A^{o}_{KN} = -\frac{aQ \cos(\theta)}{r}. 
\end{equation}
Changing the axis of rotation we can easily guess from $A_{KN}^o$ 
  a particular solution of the inhomogeneous equation 
 (\ref{eqa1}) for the arbitrary $S$ given  in (\ref{S1}):
\begin{equation}\label{akn}
A_{KN}=    -\frac{Q}{r}\; \; \sqrt{ \frac{4 \pi}{3}}  \sum_{m=1}^3 j^{(m)} S_{(\ell=1,m)}.
\end{equation}
This corresponds to a slowly rotating  KN(A)dS black hole with angular momentum components 
$ j^{(m)}$, for which
\begin{equation} \label{zkn}
\mathcal{Z}_{KN}= \frac{6Mr-4Q^2}{r^5}  \; \; \sqrt{ \frac{4 \pi}{3}}  \sum_{m=1}^3 j^{(m)} S_{(\ell=1,m)}.
\end{equation}
The general solution  of (\ref{eqa1}) is obtained by adding to the particular solution (\ref{akn}) the general 
solution of the {\em homogeneous}
 equation (\ref{eqa1}) :
\begin{equation} \label{eqa1h}
   \td^a \td_a  A^{(\ell=1)}_h  -\left(  \frac{2}{r^2}+\frac{4Q^2}{r^4}\right) A^{(\ell=1)}_h = 0.
\end{equation}
We recognize that this is the $\ell=1$ analogue of equation (\ref{mas2}), then we introduce 
\begin{equation} \label{W1}
W^{(\ell=1)}:= \frac{A^{(\ell=1)}_h}{r}
\end{equation}
as in the $\ell>1$ case 
 and, using  equations (\ref{mas2}) and (\ref{master2}), we find that (\ref{eqa1h})  is  equivalent to 
\begin{equation} \label{master2-1}
\nabla_{\a} \nabla^{\a} W^{(\ell=1)} +  \left( \frac{2M}{r^3}- \frac{6Q^2}{r^4}- \frac{2 \L}{3} \right) W^{(\ell=1)}
= 0.
\end{equation}
The solution of this equation is 
\begin{equation}\label{gshe}
r  W^{(\ell=1)} =A^{(\ell=1)}_h =  \sqrt{ \frac{4 \pi}{3}} \sum_{m=1}^3 A^{(m)}_h(t,r)  \; S_{(\ell=1,m)}(\theta,\phi),
\end{equation}
where each of the $A^{(m)}_h(t,r)$  satisfy the 1+1 wave equation (\ref{eqa1h}) 
which, introducing  a  {\em tortoise} radial coordinate defined by 
\begin{equation}
r^* = \int^r \frac{dr'}{f(r')}, 
\end{equation}
is equivalent to 
\begin{equation}\label{rw1h}
(\p_t^2 - \p_{r^*}^2 + V) A^{(m)}_h=0, \;\;\; V= f \; \left(  \frac{2}{r^2}+\frac{4Q^2}{r^4}\right).
\end{equation}

Adding (\ref{gshe}) to (\ref{akn}) gives the general solution to  (\ref{eqa1}) for the choice (\ref{S1}), this has to be 
inserted into (\ref{za})  to obtain $\mathcal{Z}$.\\

\noindent
Summarizing:\\

\begin{enumerate}
\item The $\ell=1$ gauge invariant fields  are $\mathcal{Z}$ and $A^{(\ell=1)}$. 
 The general solution of the 
$\ell=1$ LEME equations are parametrized by: i)  three constants $j^{(m)}$ 
that give $S$ (see equation (\ref{S1})) and the particular solution $A_{KN}$ of (\ref{eqa1})  
given in (\ref{akn}) and ii) 
 three solutions $A^{(m)}_h(t,r)$ of  (\ref{rw1h}) which 
span $A^{(\ell=1)}_h$ (see(\ref{gshe})). Using these gives 
\begin{equation} \label{A1gral}
A^{(\ell=1)}=A^{(\ell=1)}_h+A_{KN} =    \sqrt{ \frac{4 \pi}{3}} \sum_{m=1}^3 \left( A^{(m)}_h(t,r)  -\frac{Q}{r} j^{(m)}\right)
 \; S_{(\ell=1,m)}(\theta,\phi), 
\end{equation}
and then $\mathcal{Z}$ is obtained using  (\ref{za}), (\ref{S1}) and (\ref{A1gral}):
\begin{equation}
\mathcal{Z} = \sqrt{ \frac{4 \pi}{3}}  \sum_{m=1}^3 \left( \frac{4Q}{r^4}A^{(m)}_h(t,r)
+ \frac{6Mr-4Q^2}{r^5}   j^{(m)}
  \right)  \; S_{(\ell=1,m)}(\theta,\phi).
\end{equation}
Note that the $j^{(m)}$ in (\ref{A1gral}) are well defined: 
if we assumed that the coefficients of the harmonics of $A^{(\ell=1)}$ in (\ref{A1gral}) can be split in 
two different ways, say
$$ A^{(m)}_h(t,r)  - j^{(m)} Q/r= \tilde A^{(m)}_h(t,r)  -\tilde j^{(m)} Q/r,$$ 
this would imply that $ (\tilde j^{(m)}-j^{(m)}) Q/r$ is a solution of the homogeneous equation 
(\ref{eqa1}), which is false unless $\tilde j^{(m)}=j^{(m)}$ and thus  $\tilde A^{(m)}_h(t,r)= A^{(m)}_h(t,r)$.\\
\item In a gauge where $h_r^{(\ell=1)}=0$, we have 
 $\mathcal{Z}= \p_r (h_t^{(\ell=1)}/r^2)$, then 
\begin{equation} \label{ha-1-g}
h_a ^{(\ell=1)} dx^a =   dt \; r^2 \int^r \mathcal{Z} \; dr = \sqrt{ \frac{4 \pi}{3}} \sum_{m=1}^3 \left[  (f(r)-1) j^{(m)}  + 4Q r^2  B^{(m)}(t,r) \right] 
  S_{(\ell=1,m)} \; dt,
\end{equation}
where the $B^{(m)}$ are any three functions of $(t,r)$ such that $\p_r  B^{(m)} = r^{-4} \; A^{(\ell=1,m)}_h$ 
 (the ambiguity in the $B^{(m)}$'s gives a term $g(t)r^2 dt$ in $h_a dx^a$ which is pure gauge.) 
\end{enumerate}
It is important to note that our results are consistent with the  black hole uniqueness theorems, 
which state that any asymptotically flat stationary axi-symmetric (electro)-vacuum black hole
 is a member of the Kerr-Newman 
family. For perturbations around a Schchwarzschild black hole, $A \equiv 0$ and $Q=0$, so the $\ell=1$ equation 
(\ref{divf-l=1}) is void and the remaining equations give $\mathcal{Z}= 6M S(\theta,\phi)/r^4$ 
and then $h_a dx^a \propto \sum_m (f(r)-1)j^{(m)} S_{(\ell=1,m)}(\theta,\phi) +$ gauge terms (see (\ref{ha-1-g})), 
which corresponds to 
a slowly rotating Kerr black hole, as expected. In the $Q \neq 0$ case, however, we must rule out the existence of 
time independent solutions of the homogeneous equation (\ref{eqa1h}) (equivalently, equation (\ref{rw1h})) that 
behave properly at the horizon and for large $r$,  
to guarantee that the only time independent $\ell=1$  solution is $A=A_{KN}$ and 
 $\mathcal{Z}=\mathcal{Z}_{KN}$. Assume on the contrary that there is  a 
well behaved time independent 
solution $A(r)$ of equation (\ref{rw1h}):
\begin{equation} \label{ss}
f U A=   \p_{r^*}^2 A = f\p_r(f \p_r A), \;\;\; U=\frac{2}{r^2}+\frac{4Q^2}{r^4}.
\end{equation}
Let $r=r_h$ be the horizon radius, then for $r \simeq r_h$, $f = 2 \kappa (r-r_h) + \mathcal{O}((r-r_h)^2)$, where $\kappa>0$ is 
the surface gravity. Inserting this in (\ref{ss}) gives, for the two dimensional local solution space near $r=r_h$, 
\begin{equation}\label{series}
A = \alpha \left[ 1 + \frac{1}{\kappa r_h^2} \left(1+\frac{ 2Q^2}{r_h^2}\right) (r-r_h) + ...\right] 
+ \beta \left[ \ln\left(\frac{r-r_h}{r_h}\right) + ... \right]
\end{equation}
If $A$ is well behaved at the horizon then $\beta=0$. 
This implies (without loss of generality we assume  that $\alpha>0$) 
that at a point for $r_o>r_h$ sufficiently close to $r_h$, both $A>0$ and $\p_r A>0$ (see (\ref{series})).
Thus $\p_{r^*} A >0$ at this large negative $r^*$ value $r^*(r_o)$ and integrating equation (\ref{ss}) from $r^*(r_o)$ to 
the right and noting that $U>0$, 
we learn that $A, \p_{r^*} A $ and $\p_{r^*}^2 A$ are all 
positive for $r^*>r^*(r_o)$ and so $A \to \infty$ as $r^* \to \infty$. 
This means that time independent solutions of (\ref{rw1h})  
that are well behaved at the horizon diverge for large $r^*$. Therefore, the only acceptable stationary
$\ell=1$ solution of the LEME is then $A=A_{KN}$ and 
 $\mathcal{Z}=\mathcal{Z}_{KN}$, as we wanted to prove.

\section{Non-modal linear stability for odd perturbations} \label{nmsSect}

From the results of the previous Section follows that, introducing the field 
\begin{equation}\label{wtotal}
W :=  W^{(\ell=1)} + W^{>1}  \in L^2(S^2)_{>0},
\end{equation}
we may recast (\ref{master}), (\ref{master2}) and   (\ref{master2-1}) as the following system 
of equations for the fields $\Phi \in L^2(S^2)_{>1}$ and $W \in L^2(S^2)_{>0}$:
\begin{align} \label{Master}
&\nabla_{\a} \nabla^{\a}  \Phi  + \left( \frac{8M}{r^3}- \frac{6Q^2}{r^4}- \frac{2 \L}{3} \right) \Phi
= \frac{4Q}{r^3} \; W^{>1},\\
 \label{Master2}
&\nabla_{\a} \nabla^{\a} W +  \left( \frac{2M}{r^3}- \frac{6Q^2}{r^4}- \frac{2 \L}{3} \right) W
= - \frac{Q}{r^3} (\j2+2)\Phi,
\end{align}
 It also follows that 
the set $\mathcal{L}_-$ of odd solutions $(h_{\a \b},\f_{\a \b})$ of the LEME (\ref{1})-(\ref{3}) mod gauge equivalence,  
equation (\ref{L-}),  can be parametrized 
by the three 
gauge invariant constants $j^{(m)}$ and 
   and the gauge invariant fields  $\Phi \in L^2(S^2)_{>1}$ and $W\in L^2(S^2)_{>0}$,   subject to the  system of equations 
(\ref{Master}) and (\ref{Master2}): 
\begin{equation} \label{L-2}
\mathcal{L}_- = \{ (j^{(m)},  \Phi, W) \; | \; \text{ equations (\ref{Master}) \text{ and }(\ref{Master2}) hold } \}.
\end{equation}
This parametrization of $\mathcal{L}_-$ is interesting because is given in terms of 
gauge invariant constants and scalar fields satisfying wave equations. 
There is, however,  a distinction 
 between the constants $j^{(m)}$, which are the components of the infinitesimal angular 
momentum corresponding to perturbations within the Kerr-Newman (A)dS  family, 
and the scalar fields 
$\Phi$ and $W$, which, although  convenient as potentials to solve the $\ell>1$ LEME,   have no direct 
physical interpretation. \\ 

We will prove in Section \ref{invs} that  there are two gauge invariant, 
 physically meaningful scalar fields $\q$ and $\mathcal{F}$, that are 
directly associated to the effects of the perturbation on the curvature and on the strength of the Maxwell field, 
 and contain the same information as $ (j^{(m)}, \Phi, W)$. 
These fields accomplish the first objective 
of the nonmodal approach. 
\\

The second goal of the nonmodal approach  is to show that, if
 $\Lambda \geq 0$,  the chosen fields $\q$ and $\f$ are bounded on the outer static region by constants that depends on the initial data 
of the
perturbation on a Cauchy surface. This  makes precise the notion of nonmodal linear stability. 
To prove the pointwise boundedness we use the system of differential equations satisfied by $\q$ and $\f$,  
but we need to constrain the generality of solutions of the LEME and limit ourselves 
to the  case were  perturbation theory makes sense, which is when  perturbations  preserve the asymptotic flatness (if $\Lambda=0$) 
or de Sitter character (if $\Lambda>0$) of the background. No boundedness result is to be expected if we do not proceed so. 
Imagine, e.g.,  
that  in the $\Lambda=0$ case we take 
 initial data  $(\Phi, \dot \Phi)$ and $(W,\dot W)$ for the system  (\ref{master}) (\ref{master2}) on a $t$ slice such that 
$\Phi$ 
grows arbitrarily for large $r$. On one hand, there could be no pointwise boundedness result  on the outer static region 
for such  perturbation, 
on the other hand, the associated metric perturbation would spoil  asymptotic flatness. Treating it 
as a perturbation would be inconsistent since 
 the ``smallness'' of  $\varepsilon$ in $g_{\a \b} + \varepsilon h_{\a \b}$ would be overcome for large $r$ by the 
growth of $h_{\a \b}$. 
Thus, decay properties for large $r$ in the $\Lambda=0$ case ($r \to r_c$ if $\Lambda>0$) must be imposed on the initial data. \\
For simplicity, and to avoid complicated statements (which would inevitably  involve separate conditions for $\Lambda=0$ and $\Lambda>0$), 
we will,  following \cite{Dotti:2016cqy} and \cite{Kay:1987ax}, restrict our considerations to 
 perturbations compactly supported away from $r= \infty$ 
if $\Lambda=0$ ($r= r_c$ if $\Lambda>0$). This restriction  should not be an obstacle to generalize to milder decay conditions 
(see, e.g. the proof of Theorem 6 in \cite{Dotti:2016cqy}), and it serves our purposes of generalizing   
the results in  \cite{Dotti:2016cqy} to odd perturbations of charged black holes. 
Of course, the $r$ extent of the field support 
will  grow with $t$ as the perturbations 
evolves.\\
 The scalar fields $\q$ and $\f$ might grow high in small regions without compromising energy conservation. 
We will show, following \cite{Kay:1987ax}, that this is not the case, and that 
  it is possible to place pointwise bounds on $\q$ and $\f$ 
in the outer static region, establishing in this way the nonmodal stability of this region. 

\subsection{Measurable effects of the perturbations}  \label{invs}

Consider  the first order perturbation fields 
\begin{equation} \label{fopf}
\q =\delta (\tfrac{1}{48} C^*_{\alpha \beta \gamma \delta} C^{\alpha \beta \gamma \delta}), \;\;\;
\mathcal{F} =\d (F^*_{\alpha \beta} F^{\alpha \beta}),
\end{equation}
where $\d$ stands for derivative at $\e=0$ for a mono parametric family of solutions of the Einstein Maxwell equations,  as in equations 
(\ref{lf})-(\ref{3}),  the $\e=0$ solution being (\ref{rn})-(\ref{max}),  and a  star denotes Hodge dual
\begin{equation}
F^*_{\alpha \beta}= \tfrac{1}{2} \e_{a \b \g \d} F^{\g \d}, \;\;\; 
C^*_{\alpha \beta \g \d}= \tfrac{1}{2} \e_{a \b \mu \nu} C^{\mu \nu}{}_{\g \d},
\end{equation}
$\e_{a \b \mu \nu}$ being the volume form. Since $C^*_{\alpha \beta \gamma \delta} C^{\alpha \beta \gamma \delta} 
= F^*_{\alpha \beta} F^{\alpha \beta}=0$ in the background, the fields $\q$ and $\mathcal{F}$ are  gauge invariant 
\cite{Dotti:2016cqy} and thus suitable to analyze the effects of the perturbation in the geometry and the 
electromagnetic field. The obvious advantage of   {\em scalar} fields over  higher rank tensor fields,  
is that for the latter  there is no entirely natural concept of being ``large'' or ``small'' in a Lorentzian manifold, 
and we need this notion 
to quantify  the strength of the perturbation.\\ 

It follows from equations (\ref{pert1}), (\ref{ha}), (\ref{wfi0}),  (\ref{W1}), (\ref{A1gral})  and (\ref{fopf}),  that $\f$ 
depends on up to two derivatives of $W$ whereas 
$\q^{>1}$ depends 
on up to four derivatives of $\Phi$. However, using repeatedly the  LEME (\ref{Master})-(\ref{Master2}) and 
calculating 
separately the $\ell=1$ contributions   to  $\f^{(\ell=1)}$ and $\q^{(\ell=1)}$ coming from 
(\ref{akn})-(\ref{zkn}),  
 we can simplify considerably the resulting 
expressions 
and find, with the help of symbolic manipulation programs \cite{grtensor},  that, {\em for solutions of the LEME}, there is 
a simple relation between $\f$ and $\q$ on one side,  and $\Phi, W$ and the $j^{(m)}$ on the other: 
\begin{equation} \label{F}
\mathcal{F} =\frac{8Q^2}{r^5} \sqrt{ \frac{4 \pi}{3}}  \sum_{m=\pm1,0} j^{(m)} S_{(\ell=1,m)}+ 
 \frac{4Q}{r^3} \,\j2 W, 
\end{equation}
and
\begin{equation} \label{Q}
\mathcal{Q} = \frac{2(Q^2-Mr)}{r^6} \left[  \frac{(3Mr-2Q^2)}{r^3} \sqrt{ \frac{4 \pi}{3}}  \sum_{m=\pm1,0} j^{(m)} S_{(\ell=1,m)}+ 
 \left( \frac{\j2(\j2+2)}{4} \Phi - \frac{Q}{r} \j2 W\right) \right].
\end{equation}
The above equations allow us to prove that $\q$ and $\f$ contain all the gauge invariant information of a given  
perturbation, and  that they satisfy a coupled system of wave equations. 

\begin{thm} Consider the set of odd solutions $(h_{\a \b},\f_{\a \b})$ of the LEME (\ref{1})-(\ref{3}) around a Reissner-Nordstr\"om  
(A)dS black hole background and the set of  perturbed fields $(\f,\q)$ defined in (\ref{fopf}):
\begin{itemize}
\item[(i)] The map $([h_{\a \b}], \mathcal{F}_{\a \b}) \to (\f,\q)$ is injective: it is possible to reconstruct $\f_{\a \b}$ and 
a representative 
of $[h_{\a \b}]$ from $(\f,\q)$. 
\item[(ii)] Let \begin{equation} \label{KF}
\mathcal{K} = \left( \frac{2r^6 \q}{Q^2-Mr}+r^2 \f \right).
\end{equation}
The gauge invariant scalar fields $\f$ and $\q$ satisfy the system of wave equations
\begin{align} \label{Inv1}
&\left[ \nabla_{\a} \nabla^{\a}    + \left( \frac{8M}{r^3}- \frac{6Q^2}{r^4}- \frac{2 \L}{3} \right) \right] \mathcal{K}
= (\j2+2) \f\\ \label{Inv2}
& \left[ \nabla_{\a} \nabla^{\a}  +  \left( \frac{2M}{r^3}- \frac{6Q^2}{r^4}- \frac{2 \L}{3} \right) \right] \left(\frac{r^3}{4Q} \f \right)  
= - \frac{Q}{r^3} \mathcal{K}
\end{align}
\item[(iii)] Let $\tilde{\mathcal{K}}$ and $\tilde{\f}$ be $\ell \geq 1$ scalar fields satisfying (\ref{Inv1}) and (\ref{Inv2}), and  $\tilde{\q}= (Q^2-Mr) 
(\tilde{\mathcal{K}}-r^2\tilde{\f})/(2r^6)$ 
(c.f. equation (\ref{KF})). There exists an $\ell \geq 1$ 
 solution $(\tilde{h}_{\a \b}, \tilde{F}_{\a \b})$ of the LEME for which $\q$ and $\f$ defined in 
 (\ref{fopf}) respectively agree with $\tilde{\q}$ and $\tilde{\f}$.
\end{itemize}
\end{thm}
\begin{proof}
\mbox{}
\begin{itemize}
\item[(i)] Expand all fields in the orthonormal 
 basis of spherical harmonics $S_{(\ell,m)}$, e.g., $\f = \sum_{(\ell,m)} \f^{(\ell,m)} S_{(\ell,m)}$ 
(then $\f ^{(\ell=1)}= \sum_{m=1}^3 \f^{(\ell=1,m)} S_{(\ell=1,m)}$ and $\f^{>1} = \sum_{(\ell>1,m)} \f^{(\ell,m)} 
S_{(\ell,m)}$) 
and similarly for $\q$, $\Phi$ and $W$. 
Recall that $\f, \mathcal{Q},  W \in L^2(S^2)_{>0}$ whereas $\Phi \in L^2(S^2)_{>1}$.
From equations (\ref{W1}), (\ref{A1gral}), (\ref{wtotal})  and (\ref{F}) follows that 
\begin{equation}
\f ^{( \ell=1)} = -8\frac{Q}{r^4} A^{( \ell=1)}.
\end{equation}
Thus, from $\f ^{( \ell=1)}$ we obtain $A^{( \ell=1)}$ which, inserted in (\ref{pert1}) gives the $\ell=1$ piece of 
the electromagnetic field perturbation and inserted in  (\ref{eqa1}) gives $S$. 
Using $A^{( \ell=1)}$ and $S$ in  (\ref{za}) gives the $\ell=1$ field $\mathcal{Z}$. In any gauge with $h_r=0$, 
$\mathcal{Z}=  \p_r (h_t^{(\ell=1)}/r^2)$, this implies that 
the $\ell=1$ piece of the metric perturbation  can be obtained by integration  (see equation (\ref{ha-1-g})). \\

To reconstruct the $\ell>1$ pieces of the fields (\ref{pert1}) we proceed as follows: 
 from (\ref{F}), $\f^{>1}= 4Q \j2 W^{>1}/r^3$, therefore  $\f_{(\ell,m)} = -(4Q/r^3) \ell (\ell+1) 
W_{(\ell,m)}$ for $\ell>1$, i.e., $\f^{>1}$ gives $W^{>1}$. Combining 
the $\ell>1$ projections of (\ref{F}) and (\ref{Q}) gives (see (\ref{KF}))
\begin{equation} \label{k>1}
\mathcal{K}^{>1} = \frac{2 r^6 \q^{>1}}{Q^2-Mr} + r^2 \f^{>1} = \j2 (\j2+2) \Phi,
\end{equation}
from where $\Phi$ can be obtained since  the operator $\j2(\j2+2)$ is invertible in $L^2(S^2)_{\ell>1}$, acting 
as $(\ell+2)(\ell+1)\ell (\ell-1)$ on any $\ell>1$ subspace of $L^2(S^2)$ (i.e., 
 $(\j2(\j2+2) \Phi)_{(\ell,m)} = (\ell+2)(\ell+1)\ell (\ell-1) \Phi_{(\ell,m)}$.)\\
Once we have  $W^{>1}$ and $\Phi$  the $\ell>1$ electromagnetic perturbation is obtained by inserting $A^{>1}= r 
W^{>1}$ in (\ref{pert1}) and the Regge-Wheeler representative of the $\ell>1$ metric perturbation is obtained 
inserting 
$Z=r^2\Phi$ in 
equations (\ref{ha}) and (\ref{pertrw}).\\

\item[(ii)] 
From (\ref{k>1}), using $[\j2,\nabla_{\a}]=0$ and the $\ell>1$
equations (\ref{master})-(\ref{master2}) we find that the projections $\q^{>1}$ and 
$\f^{>1}$ satisfy  the system of equations (\ref{Inv1})-(\ref{Inv2}):
\begin{align} 
&\left[ \nabla_{\a} \nabla^{\a}    + \left( \frac{8M}{r^3}- \frac{6Q^2}{r^4}- \frac{2 \L}{3} \right) \right] \mathcal{K}^{>1}
= (\j2+2) \f^{>1}\\ 
& \left[ \nabla_{\a} \nabla^{\a}  +  \left( \frac{2M}{r^3}- \frac{6Q^2}{r^4}- \frac{2 \L}{3} \right) \right] \left(\frac{r^3}{4Q} \f^{>1} \right)  
= - \frac{Q}{r^3} \mathcal{K}^{>1}
\end{align}
The $\ell=1$ piece of of $\mathcal{K}$,
\begin{equation}
\mathcal{K}^{(\ell=1)} = \left( \frac{2r^6 \q^{(\ell=1)}}{Q^2-Mr}+r^2 \f^{(\ell=1)} \right) = \frac{12M}{r^2} \sqrt{ \frac{4 \pi}{3}}  \sum_{m=\pm1,0} j^{(m)} S_{(\ell=1,m)},
\end{equation}
together with that of $\f$ 
\begin{equation} \label{F1}
\f^{(\ell=1)} =\frac{8Q^2 }{r^5} \sqrt{ \frac{4 \pi}{3}}  \sum_{m=\pm1,0} j^{(m)} S_{(\ell=1,m)} -
 \frac{8Q}{r^3}W^{(\ell=1)},
\end{equation}
also verify  (\ref{Inv1})-(\ref{Inv2}). This can  be checked
 using equation (\ref{master2-1}) and the fact that the wave operator on the left of equation 
(\ref{Inv1}) gives zero when acting on $S_{(\ell=1,m)}/r^2$.   Thus, equations (\ref{Inv1})-(\ref{Inv2}) follow.
\item[(iii)] Define 
 \begin{equation} \label{inverse1}
\tilde  \Phi = [\j2 (\j2+2)]^{-1} \; \tilde{\mathcal{K}}^{>1}, \;\;\; \;  \tilde W^{>1} = \frac{r^3}{4Q} [\j2]^{-1} \tilde{\f}^{>1}.
\end{equation}
Equations  (\ref{Inv1}) and (\ref{Inv2}) imply that the fields (\ref{inverse1}) satisfy the system of equations (\ref{master}) and 
(\ref{master2}) and therefore,  $\tilde h_a^{>1} =\tilde \e_a{}^b \td_b(r^2 
 \tilde \Phi)$ 
and $\tilde A^{>1}=r  \tilde W^{>1}$ satisfy the $\ell \geq 1$ 
 LEME  (see the paragraph below equation (\ref{master2})). In view of (\ref{F})-(\ref{KF}) and (\ref{k>1}), 
 the associated   $\ell \geq 1$  solution class $([\tilde h_{\a \b}], F_{\a \b})$, and in particular its RW representative,  
has $\delta (\tfrac{1}{48} C^*_{\alpha \beta \gamma \delta} C^{\alpha \beta \gamma \delta})= \tilde{\q}^{>1}$ and 
$\d (F^*_{\alpha \beta} F^{\alpha \beta})= \tilde{\f}^{>1}$
\end{itemize}
\end{proof}

In the $Q \to 0$ limit equations (\ref{Inv1}) and (\ref{Inv2}) decouple. The first one gives the four dimensional Regge-Wheeler equation 
for $\q$ used in \cite{Dotti:2013uxa}  and \cite{Dotti:2016cqy} to prove the nonmodal linear stability of 
the Schwarzschild dS black hole and  the second one gives 
the Fackerrel-Ipser equation for 
a test Maxwell field on a Schwarzschild (A)dS black hole \cite{Fackerell:1972hg} \cite{Jezierski:2015lwa}.

\subsection{Pointwise boundedness of $\mathcal{Q}$ and $\mathcal{F}$ for $\Lambda \geq 0$} \label{pb}

The standard way of solving the $\ell>1$ LEME (\ref{master})-(\ref{master2}) is  by projecting this system onto 
the $\ell$ subspaces and 
then decoupling the resulting pair of fixed $\ell$ equations by introducing two Regge-Wheeler fields 
 \cite{Kodama:2003kk} \cite{Zerilli:1974ai}.
 This is 
equivalent to introducing the operator
\begin{equation}
\Xi = \sqrt{9M^2-4Q^2(\j2+2)}, 
\end{equation}
which is well defined and positive definite in $L^2(S^2)_{>0}$, as it acts  on $L^2(S^2)_{\ell}$ as multiplication 
times $$\sqrt{9M^2+4Q^2(\ell+2)(\ell-1)},$$ and  two fields $\Phi_{n}, n=1,2$ in terms of which 
\begin{align} \label{wfi}
W^{>1} &= \frac{(3M+ \Xi)}{r} \;  \Phi_1 + \frac{Q}{r} (\j2+2) \; \Phi_2\\
\Phi &= \frac{4Q}{r} \Phi_1 + \frac{(3M+\Xi)}{r}\;  \Phi_2. \label{Pfi}
\end{align}
This makes the system (\ref{master})-(\ref{master2}) equivalent to the Regge-Wheeler equations, 
first derived in \cite{Zerilli:1974ai},  
\begin{equation} \label{RWE}
(\p_t^2 - \p_{r^*}^2 + f {U_n}) \Phi_n =0, \;\; \Phi_n \in L^2(S^2)_{>1}, \; n=1,2
\end{equation}
where $r^*$  is a  {\em tortoise} radial coordinate 
and
\begin{equation} \label{rwp}
U_n= -\frac{\mathbf{J}^2}{r^2} + \frac{4Q^2}{r^4} - \frac{3M+(-1)^n \Xi}{r^3}.
\end{equation}
In terms of these fields, the $\ell>1$ piece of $\q$ and $\f$ are
\begin{align} \label{q>1}
\q^{>1} &= \frac{2(Q^2-Mr)}{r^6} \left(\frac{Q}{r} \left[(\j2+2)-\frac{3M+\Xi}{r}\right] \j2 \Phi_1 + \left[ \frac{3M+\Xi}{4r}-\frac{Q^2}{r^2}\right] \j2 
(\j2+2) \Phi_2\right)\\ \label{f>1} 
\f^{>1} &= \frac{4Q}{r^4} (3M +\Xi) \j2 \Phi_1 + \frac{4Q^2}{r^4} (\j2+2)\j2 \Phi_2
\end{align}

Up to this point, the considerations in this paper were insensitive to the value of the cosmological constant: 
odd perturbations can always be treated using the gauge invariant potentials $\Phi$ and $W$ and  constants 
$j^{(m)}$, and Theorem 1 holds irrespective of the value of $\Lambda$. 
In the rest of this Section, however, we will  consider  the {\em evolution of 
initial data} for the LEME, for which we  need to make a distinction between the cases $\Lambda <0$ and $\Lambda \geq 0$ 
due some key differences  in their causal structure. \\
In the asymptotically AdS case $\Lambda<0$, $f$ in (\ref{f}) has two positive roots $0 < r_i < r_h$ (we will restrict  for the moment 
to the non extremal case $r_i \neq r_h$) and the hypersurfaces 
they define bound three regions: I ($0<r<r_i$), II ($r_i<r<r_h$) and III ($r_h<r$). Isometric copies of these regions are obtained 
by ``Kruskalizing'' around the simple roots $r_i, r_h$ of $r^2 f$. This gives the maximal analytic extension 
depicted in figure \ref{adsfig}, which extends infinitely in the vertical direction. Note that regions I and III, where $f>0$,  are static whereas $f<0$ 
in region II, which is therefore non static. Note also  that the union of regions 
II, II', III and III'   fails to be globally hyperbolic due to the 
timelike character of the future and past null infinities $\mathcal{I}^{\pm}$. This  is the peculiar aspect of  asymptotically anti de Sitter spaces
that differentiates it from asymptotically de Sitter or flat spaces. 
In the asymptotically AdS case the dynamics of wave-like equations  requires a prescription of boundary conditions at the conformal 
timelike boundary  $\mathcal{I}^- \cup \mathcal{I}^+$, which corresponds to $r=\infty$. Different boundary conditions lead to different dynamics, 
including unstable and stable ones \cite{bernardo}. For this reason, from now on, we restrict to the cases $\Lambda \geq 0$, 
for which the dynamics is unique and, as we will show, stable.\\

\begin{figure}[htb]
    \includegraphics[width=.3\textwidth]{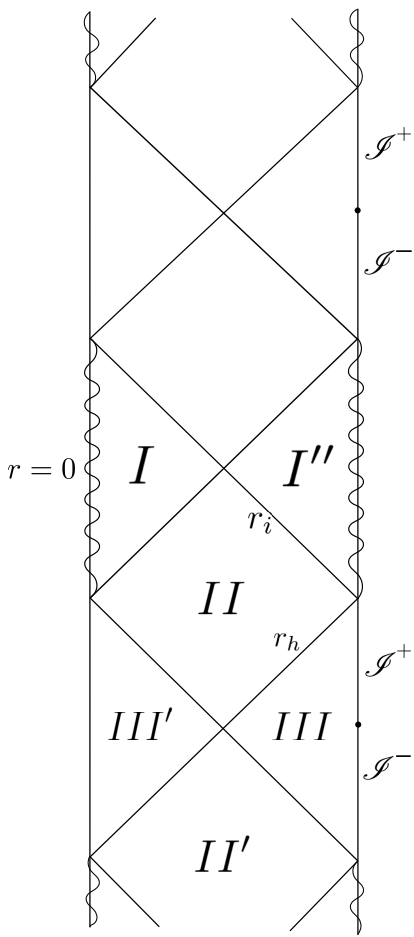}
  \caption{\label{adsfig} The Carter-Penrose diagram of (part of) the maximal analytic extension of the \rn AdS black hole. 
The union of II, II', III and III'   fails to be globally hyperbolic due to the 
timelike character of $\mathcal{I}^- \cup \mathcal{I}^+$,  $\mathscr{I}^- \cup \mathscr{I}^+$ which is peculiar to asymptotically anti de Sitter spaces.}
\end{figure}

For $\Lambda=0$ and $|Q| <M$, $f$ has  again two positive roots $r_i<r_h$ that correspond  respectively 
to the Cauchy and black hole 
horizons.  In this case  $f= (r-r_i)(r-r_h)/r^2$ with  $Q^2=r_i r_h$ and $M=\tfrac{1}{2}
(r_i+r_h)$. 
The outer static region, region III in Figure \ref{rnfig},  corresponds to $r>r_h$ whereas the inner static 
region $I$ is the one defined by
$0<r<r_i$; the singularity at $r=0$ is covered by these two horizons. Kruskalizing at $r_i$ and $r_h$ 
we get further copies of these regions resulting  the diagram in the figure, which  extends 
infinitely in the vertical direction. 
The union of II, II', III and III' 
is globally hyperbolic, I and I' being extensions beyond the Cauchy horizon at $r=r_i$, which is the future boundary 
of the maximum Cauchy development of initial data given at a complete spacelike hypersurface extending  from spacelike infinity 
in region III' to spacelike infinity 
in region III. In the extreme case $|Q|=M$, $r_i=r_h$ and region II collapses. For $|Q|>M$ the spacetime 
is not a black hole but an (unstable, see \cite{Dotti:2006gc}  and \cite{Dotti:2010uc}) naked singularity. 

\begin{figure}[htb]
    \includegraphics[width=.3\textwidth]{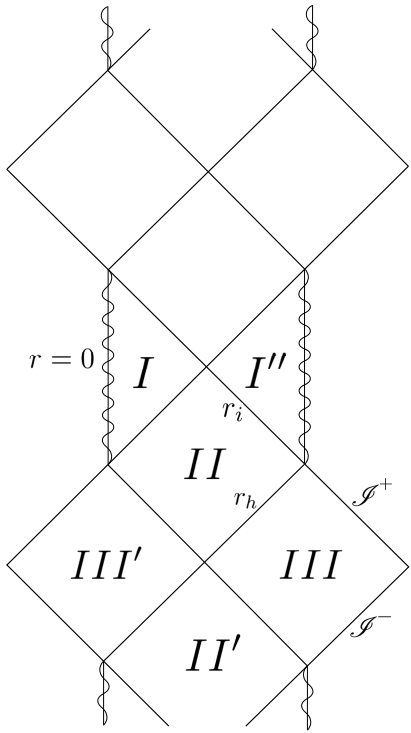}
  \caption{\label{rnfig} The Carter-Penrose diagram of (part of) the maximal analytic extension of the $|Q|<M$ 
\rn  black hole. 
The union of II, II', III and III'   is globally hyperbolic, its boundary at $r_i$ is a Cauchy horizon.}
\end{figure}

 For $\Lambda>0$ we focus 
on the non extremal cases, for which  $f$ has three simple positive roots $0<r_i<r_h<r_c$ which correspond to  the inner, black hole and cosmological 
horizons respectively, and  a fourth root at $r=-(r_i+r_h+r_c)$:
\begin{equation} \label{fr}
f = - \frac{(r-r_i)(r-r_h)(r-r_c)(r+r_i+r_h+r_c)}{r^2(r_i^2+r_h^2+r_c^2+r_i r_h+r_i r_c+r_hr_c)}, 
\end{equation}
$Q^2, M$ and $\Lambda$, as well as the  relations among them 
can be found  in terms of $r_i, r_h$ and $r_c$  by comparison  of (\ref{fr}) with  (\ref{f}).
 As before, regions separated by the horizons are numbered in increasing number for larger $r$ values. 
Since we can Kruskalize around all three horizons and large $r=$ constant hypersurfaces are spacelike, 
we get the  diagram in figure \ref{rndsfig}, which extends infinitely in both directions. 
There are a number of extremal cases corresponding to $r_i=r_h$, $r_h=r_c$, etc, the Carter-Penrose diagrams
for these cases can be found in \cite{lake}.\\

\begin{figure}[htb]
    \includegraphics[width=.6\textwidth]{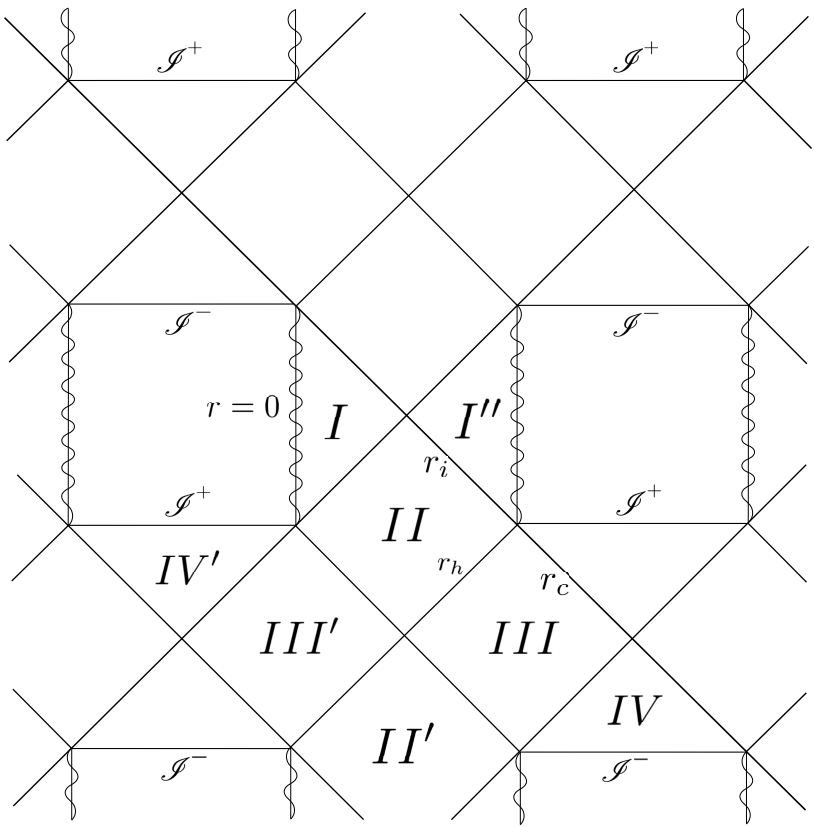}
  \caption{\label{rndsfig} The Carter-Penrose diagram of (part of) the maximal analytic extension of a non extremal 
(three different horizons) 
\rn  de Sitter black hole. }
\end{figure}

In what follows,  we will prove the stability of the outer static region III of 
$\Lambda \geq 0$ \rn black holes. To this purpose, we will consider the 
union of regions II, II', III and III', which is  globally hyperbolic,  and  study the evolution of perturbations from data 
on a Cauchy surface. Any Cauchy surface has two ends, one at each copy of spacelike infinity (if $\Lambda=0$) 
or the $r_c$ bifurcation sphere (if $\Lambda>0$) in regions III and III'. 
As explained above,  we will restrict our considerations to perturbations with initial data compactly supported away from these  ends.  
Relevant perturbations can be more general, as long as they  preserve the asymptotically flat (AdS) character of the background, however,  
for the seek of simplicity and to allow a unified treatment of the $\Lambda=0$ and $\Lambda>0$ case 
we will assume compact support, as  in  \cite{Kay:1987ax}.\\

Following \cite{Kay:1987ax} we  write (\ref{RWE})-(\ref{rwp}) as 
\begin{equation} \label{RWE2}
\p_t^2 \Phi_n+A_n \Phi_n=0
\end{equation}
where
\begin{equation}
A_n=- \p_{r^*}^2 + V_1 + V_2(-\mathbf{J}^2) + {}_n{V_3}\;\Xi \label{Aeq}
\end{equation}
and
\begin{equation} \label{v123}
V_1=f\Big(\frac{4Q^2}{r^4} - \frac{3M}{r^3}\Big) \;, V_2=\frac{f}{r^2}  \;,\;  {}_nV_3=(-1)^n \frac{f}{r^3}
\end{equation}
are bounded functions on the outer static region $III$ for $\Lambda \geq 0$. 
A non trivial fact, proved in Section 6.2.1 of \cite{Kodama:2003kk} is that the $A_n, n=1,2$ are positive definite self adjoint operators in 
the space $L^2(\mathbb{R} \times S^2, dr^*\sin(\theta) d\theta d\phi)$  of square integrable functions of region $III$ 
under this particular measure. The proof is based on a particular $S-$deformation (defined in \cite{Kodama:2003kk}) of  the $A_n$'s. 
\\

The proof of the following Theorem is a straightforward adaptation to equation (\ref{RWE}) of Theorem 1  in \cite{Kay:1987ax} 
which is about  the Klein Gordon 
equation 
 on a Schwarzschild background . It  uses   the self adjointness and positive character 
of the $A_n$, $-\j2$ and $\Xi$:

\begin{thm}\label{bound}  Assume $\Phi_n$ is a solution of equation (\ref{RWE}) on the union of regions II, II', III and III' 
of the  extended \rn (figure \ref{rnfig} or \rn   de Sitter (figure \ref{rndsfig}) spacetimes,  
which has compact support on Cauchy surfaces. There exists a constant $C$ that depends on the 
datum of this field at a Cauchy surface, such that $|\Phi_n|<C$ for all points in the outer static region $III$. 
\end{thm}
\begin{proof}
The argument in \cite{Kay:1987ax} showing that we can restrict to fields that vanish at the bifurcation sphere together with 
its Kruskal time derivative holds here  case because the $Z_2$ required isometry exchanging III $ \leftrightarrow$ III'  is also available 
in this case. 
This implies that we may restrict our attention  to fields in the outer static region decaying towards the bifurcation sphere as 
detailed in the Appendix in \cite{Kay:1987ax}. \\
 
On a  $t$ slice of region III, define the $L^2$ norm of a real field $G$ as 
\begin{equation}
\lVert G \rVert^2=\langle G|G\rangle=\int_{\mathbb{R} \times S^2}  G^2 \;  dr^* \sin(\theta) d\theta d\phi, \;\; dr^* = \frac{dr}{f}.
\end{equation}
Note that this is {\em not} the volume element  induced from the spacetime metric. The usefulness 
of the above norm lies in 
the Sobolev type inequality, equation (5.27)  in  \cite{dk}, relating it with a point wise boundedness of $G$ on the slice
\begin{equation}
|G(r^*,\theta,\phi)| \leq  K \left( \lVert G \rVert + \lVert \p_{r^*}^2 G \rVert +
\lVert \j2 G \rVert \right) \;, \;\; (r^*,\theta,\phi) \in \mathbb{R} \times S^2 
\end{equation}
where $K$ is a constant. Applying this to  the Regge-Wheeler fields $\Phi_n$ at a fixed time $t$ gives 
\begin{equation}
|\Phi_n(t,r^*,\theta,\phi)| \leq  K \left(\left. \lVert \Phi_n \right|_t \rVert + \left. \lVert \p_{r^*}^2 \Phi_n \right|_t \rVert +
\left. \lVert \j2 \Phi_n \right|_t \rVert \right). \label{sobo2}
\end{equation}
As in the Appendix in \cite{Kay:1987ax}, we will follow the  strategy of proving that the $L^2$ norms 
on the right hand side of (\ref{sobo2}) can be bounded by the energies of related field configurations. Since 
energy is conserved for solutions of (\ref{RWE2}), we get in this way a $t-$independent upper bound of the 
right side of (\ref{sobo2}) and therefore,  a global bound of $|\Phi_n(t,r^*,\theta,\phi)|$ for all $(t,r^*,\theta,\phi)$, 
i.e., of $\Phi_n$ in the outer static region III. \\

The conserved (i.e., $t-$independent) energy associated to equation (\ref{RWE2}) is 
\begin{equation} \label{energy}
E = \frac{1}{2} \int _{\mathbb{R} \times S^2} ((\p_t \Phi_n)^2 + \Phi_n A_n \Phi_n )\;  dr^* \sin(\theta) d\theta d\phi.
\end{equation}
Since $E$ does not depend on $t$, we may (and will) regard it as a functional on the initial datum: $E = E(\Phi^o_n, \dot \Phi^o_n)$, where $\Phi^o_n
 = \left.  \Phi_n \right|_{t_o}$ and $ \dot \Phi^o_n=\left. (\p_t \Phi_n) \right|_{t_o}$:
\begin{equation}
E(\Phi^o_n, \dot \Phi^o_n)= \frac{1}{2} \int _{\mathbb{R} \times S^2} ((\dot \Phi^o_n)^2+ \Phi^o_n A_n \Phi^o_n
 )\;  dr^* \sin(\theta) d\theta d\phi.
\end{equation}
From (\ref{Aeq})  
\begin{equation}
\lVert \p_{r^*}^2 \fn \rVert \leq \lVert A_n \fn \rVert + \lVert V_1 \rVert_\infty \lVert \fn \rVert +
\lVert V_2 \rVert_\infty \lVert \j2 \fn \rVert + \lVert {}_nV_3 \rVert_\infty \lVert \Xi \fn \rVert,
\end{equation}
where $\lVert V_1 \rVert_\infty$ is the least upper bound  of  
 $V_1$  on the outer region III,  and similarly for the other terms. Combining this with 
(\ref{sobo2}) gives 
\begin{equation}
|\Phi_n(t,r^*,\theta,\phi)| \leq  K' \left(\left. \lVert \Phi_n \right|_t \rVert +  \lVert A_n \fn \rVert + 
\left. \lVert \j2\Phi_n \right|_t \rVert + 
\lVert \Xi \fn \rVert
\right). \label{sobo3}
\end{equation}
We now use the fact that applying to a Cauchy datum  $(\Phi_n^o, \dot \Phi_n^o)$ an operator that is a function of $\j2$ or 
$A_n$ commutes with time evolution \cite{Kay:1987ax}, and also use the positive definiteness of the $A_n$ to define 
$A_n^{\pm 1/2}$ by means of  the spectral theorem. This allows us   to estimate 
each term on the right hand side of (\ref{sobo3}) with the energy of field configurations related to the one with initial datum 
$(\Phi_n^o, \dot \Phi_n^o)$ (note that the first three equations below are taken verbatim from the Appendix in \cite{Kay:1987ax}),
\begin{align}
& \lVert \fn \rVert^2  \leq \lVert \left. \Phi_n \right|_t\rVert^2 + \lVert A_n^{-\frac{1}{2}}\left. \dot 
\Phi_n \right|_t \rVert^2 =
2 \;E\left(A_n^{-\frac{1}{2}} \Phi_n^o,A_n^{-\frac{1}{2}} \dot \Phi_n^o\right),\\
&   \lVert A_n \fn \rVert^2 \leq   \lVert A_n \fn \rVert^2 + \lVert A_n^{\frac{1}{2}} \dfn\rVert^2 = 2 \; E\left(A_n^{\frac{1}{2}} \Phi_n^o,A_n^{\frac{1}{2}} \dot \Phi_n^o\right),\\
&  \lVert \j2 \fn  \rVert^2 \leq   \lVert \j2 \fn \rVert^2 +  \lVert A_n^{-\frac{1}{2}} \j2 \dfn\rVert^2 =
 2 \; E\left(A_n^{-\frac{1}{2}} \j2 \Phi_n^o,A_n^{-\frac{1}{2}} \j2 \dot \Phi_n^o\right),\\
& \lVert \Xi \fn  \rVert^2 \leq   \lVert \Xi \fn \rVert^2 +  \lVert A_n^{-\frac{1}{2}} \Xi \dfn\rVert^2 =
 2 \; E\left(A_n^{-\frac{1}{2}} \Xi \Phi_n^o,A_n^{-\frac{1}{2}} \Xi \dot \Phi_n^o\right),
\end{align}
and to replace the right hand side of (\ref{sobo3}) with a time independent constant made out of 
the datum $(\Phi_n^o, \dot \Phi_n^o)$, as desired.
\end{proof}

\begin{cor}
Let $\f$ and $\q$ be the fields (\ref{fopf}) associated to a solution of the LEME.  
Under the assumptions of the Theorem, 
in the outer static region III of a $\Lambda \geq 0$ \rn black hole  
\begin{equation} \label{invbounds}
\f < \frac{\f_o}{r^4}, \;\; \q < \frac{\q_o}{r^6},
\end{equation}
where $\f_o$ and $\q_o$ are constants that depend on the Cauchy datum $(j^{(m)},A_h^{(m)},\dot A_h^{(m)},\Phi^o_n, \dot \Phi^o_n)$ of 
the perturbation.
\end{cor}
\begin{proof}
We use that the fields $(\j2+2)\j2 \Phi_n, \Xi (\j2+2)\j2 \Phi_2$ and $\Xi \j2 \Phi$ that appear in (\ref{q>1}) and (\ref{f>1}) all satisfy 
equation (\ref{RWE2}) and so, according to the Theorem,  are bounded by a constant. 
This implies that $\f^{>1} <\f_o^{>1}/r^4$, where the constant $\f_o^{>1}$ depends on the initial $\ell>1$ perturbation 
data, and similarly $\q < \q_o^{>1}/r^6$.\\
 To see that the $\ell=1$ pieces do not spoil these  bounds we use (\ref{F1}) and the $\ell=1$ piece  
 of (\ref{Q}):
\begin{equation} \label{Q1}
\mathcal{Q}^{(\ell=1)} = \frac{2(Q^2-Mr)}{r^6} \left[  \frac{(3Mr-2Q^2)}{r^3} \sqrt{ \frac{4 \pi}{3}}  \sum_{m=\pm1,0} j^{(m)} S_{(\ell=1,m)} + \frac{2Q}{r}  W^{(\ell=1)}  \right].
\end{equation}
Each harmonic component $A_h^{(m)}$of $A^{(\ell=1)}_h=r W^{(\ell=1)}$ 
(equation (\ref{gshe})), satisfies the 1+1 wave equation (\ref{rw1h}), which 
is of the form of equation (1) in \cite{wald} with $V$ satisfying the hypothesis used in that paper, therefore (see 
equation (19) 
in the erratum o
\cite{wald}), the $A_h^{(m)}, m=1,2,3$ are bounded, for all $t$ and $r^*$,  by constants 
that depend on the initial data 
$(A_h^{(m)},\dot A_h^{(m)})$ for these fields. 
 This gives $ W^{(\ell=1)} <$ constant$/r$ which, in view of (\ref{F1}) and (\ref{Q1}) implies 
$ \f^{(\ell=1)} <$ constant$/r^4$ and $\q^{(\ell=1)} <$ constant$/r^7$, which  is consistent with (\ref{invbounds}).
\end{proof}

Theorem 1.i together with the above Corollary prove our notion of non-modal linear stability for the outer static 
region of $\Lambda \geq 0$ \rn black holes.

\section{Cosmic censorship and related instabilities}\label{cs}

 The two isometric 
copies of the region $0<r<r_i$ attached to the future of $r_i$ are one among infinitely many different possible extensions of the spacetime 
beyond $r_i$ (although  the only {\em analytic} one). For $\Lambda \geq 0$, this extension spoils 
the global hyperbolicity of the union of regions II, II', III and III' 
 by introducing causal curves that end in the past at the  $r=0$ singularity. Regions I and I' are beyond the maximal 
Cauchy development of a spacelike surface extending from spacelike infinity (bifurcation sphere at $r_c$) in III' 
to spacelike infinity (bifurcation sphere at $r_c$) in III in the $\Lambda=0$ ($\Lambda>0$) case.  
This  is a complete  Cauchy surface if $\Lambda=0$, and 
the possibility of  smoothly  extending the maximal future development of this surface beyond its Cauchy horizon is a 
 rather disturbing feature of General Relativity, 
considered to be {\em non generic}, in a sense
 yet to be made precise, and referred to as the {\em strong cosmic censorship conjecture}, 
first 
proposed by Penrose almost fifty years ago \cite{penrose}. The original argument given by Penrose 
for the $\Lambda=0$ charged black hole,  is that a small amount of radiation 
 originating outside the black hole and 
coming into the non static region II ($r_i <r<r_h$)  
 is gravitationally blueshifted
as it propagates inwards parallel to the Cauchy horizon, in a way such that the 
energy flux measured by an observer in free fall towards (the right copy, see  figures \ref{rnfig} and \ref{rndsfig}) 
 of region I ($0<r<r_i$)   
diverges as $r \to {r_i}^+$.   This idea has proved to hold true 
for electro-gravitational perturbations at the linear level in \cite{chanhartle}, where it was shown that for a radially free falling   
observer with 4-velocity $u^a$ 
the $(\ell,m)$ piece of 
\begin{equation} \label{hartle}
\frac{d \Phi_n}{d \tau} = u^a \nabla_a \Phi_n
\end{equation}
(and therefore the complete field) 
diverges for $n=1,2$ as the observer approaches the Cauchy horizon. Although for  some time it was 
thought that a positive cosmological constant introduces a competition of red and blueshifts effects that 
 prevents this divergence in the case of nearly extremal ($r_i \lesssim r_h$) black holes 
\cite{Chambers:1997ef}, it was later proved in \cite{Brady:1998au} that if we take into account 
the  contributions of scattered outgoing modes the divergence occurs for any positive value 
of $\Lambda$ allowing for a three horizon structure. 
Further evidence of the instability of the Cauchy horizon
are   \cite{Poisson:1990eh}  and related works,  
as well as  more recent models including a  scalar field (to avoid Birkhoff's theorem), see \cite{Burko:2002fv}  and \cite{Maeda:2005yd}. \\
The Regge-Wheeler-Zerilli potentials $\Phi_n$ in equation (\ref{hartle}), whose derivative with 
respect to proper time is shown to diverge 
at the Cauchy horizon in \cite{chanhartle} for $\Lambda=0$ and for $\Lambda>0$ in \cite{Brady:1998au} are,
of course, non observable, as they are potentials for the metric and electromagnetic field perturbations, although the 
square of (\ref{hartle}) contributes to the flux of energy of the perturbation. 
However, the $\Phi_n$ enter the harmonic expansion of the $\ell>1$ pieces of $\mathcal{Q}$ and $\mathcal{F}$,
 and the implications   of the combined set of   equations (\ref{q>1}), (\ref{f>1}) and (\ref{hartle}) are immediate: 
 for a radially infalling  observer 
crossing $r_i$,  
the rate 
$dr/d\tau$ is clearly nonzero and finite. For such an observer, $d/d\tau= u^c \nabla_c$ commutes with 
the angular operators $\j2$ and $\Xi$.  
This implies that both $d \mathcal{F}_{>1}/d \tau$  and $d \mathcal{Q}_{>1}/d\tau$ will diverge 
along this geodesic as $r\to r_i^+$ (and so will diverge $d \mathcal{F}/d \tau$  and $d \mathcal{Q}/d\tau$), 
suggesting that the Cauchy horizon is replaced with a curvature singularity. 
Of course, this statement needs to be taken with care since
as soon as $r_i$ is approached and these quantities  start to grow, the linearized 
equations become useless an one needs to study the evolution of the perturbation using other techniques; 
 but in any case the divergence 
at $r_i$ of the linear fields 
 is a clear indication of strong cosmic censorship. \\

The extreme case $r_i=r_h$ has been less studied, although  it was recently shown that for $\Lambda=0$ (case in which the extreme black 
holes corresponds to $|Q|=M=r_i=r_h$) the transverse derivative $\p \Phi_n / \p r$ 
 in coordinates $(v=t+r^*,r,\theta,\phi)$ (again, for a fixed $\ell$ component) diverges at $r=r_i$ as $v \to \infty$ along 
the horizon null generators 
\cite{Lucietti:2012xr}. The divergence follows from  a set  quantities  that are shown from   (\ref{RWE}) 
to be conserved along the $r_i=r_h$ horizon generators. These are analogous 
to a similar   set of conserved quantities  for the massless scalar wave equation  found in 
\cite{Aretakis:2012ei}, the conservation of which was  shown in \cite{Bizon:2012we} to follow from a combination 
of  Newman Penrose conserved quantities at null infinity 
\cite{newmanpenrose} and a conformal discrete isometry that exchanges the degenerate  horizon and null infinity, 
isometry discovered in \cite{couch} and  used in  \cite{Dain:2012qw} 
to explain the symmetry of the effective potential and the consistency of the pointwise bounds for a massles scalar field 
in the extreme case.\\
 Perturbations of the extreme $\Lambda=0$ case were studied non linearly in a recent paper 
by Reiris \cite{Reiris:2014xca}  where it was shown that small electro-vacuum perturbations of initial data of 
extreme Reissner-Nordstr\"om black holes cannot decay in time into an extreme Kerr-Newman black hole. The evidence 
in \cite{Reiris:2014xca}  is that these non-stationary solutions of the Einstein-Maxwell equations 
 will settle into a sub-extremal black hole of the Kerr-Newman family. \\

As a final comment we mention that, besides the strong evidence supporting the idea that  a slightly perturbed 
Reissner-Nordstr\"om black hole will develop a curvature singularity that cuts off the innermost region I ($0 <r<r_i$), 
this non-unique extension beyond $r_i$ is by itself linearly unstable under electro-gravitational perturbations 
\cite{Dotti:2010uc}. The instability, confined to this region, belongs to the even sector of the linear perturbations.

\section{Discussion}

We  proved that the odd sector of Einstein-Maxwell perturbations around a \rn (A)dS black hole shares  
with  the uncharged Schwarzschild black hole the  property that there are 
 physically meaningful, gauge invariant scalar fields $\q$ and $\f$ encoding the same information 
as a gauge class of a metric perturbation and  satisfying a system of four dimensional wave equations which are 
entirely equivalent to the linearized Einstein-Maxwell equations. 
For uncharged black holes $Q=0$ the system of equations  decouple, leaving 
the equation found in \cite{Dotti:2013uxa}   \cite{Dotti:2016cqy} for the gravitational degrees of freedom,  
encoded in $\q$, and the Fackerrel Ipser equation for the Maxwell degrees of freedom $\f$. \\
Besides the significant reduction of the linearized Einstein-Maxwell system 
 to scalar field equations, the  resulting system of equations allow us 
 to prove that, for generic perturbations,  $\q$ and $\f$ are pointwise bounded in the outer static region. 
This  gives a strong notion of linear stability in this region, 
analogous to that found in \cite{Dotti:2013uxa}  and \cite{Dotti:2016cqy}  for the Schwarzschild black hole. \\

If we assume that the large $t$ decay at fixed $r$ of solutions of  the Regge-Wheeler equation in the $Q=0$ case 
(known as Price tails,  
see \cite{Price:1971fb},  and \cite{Brady:1996za} for $\Lambda>0$)
occurs also for the $Q \neq 0$ Regge-Wheeler fields $\Phi_n$ in (\ref{RWE2}), and note   that 
equation (\ref{rw1h}) for $A_h^{(m)}$ formally agrees with the Regge-Wheeler equation (\ref{RWE})-(\ref{rwp}) for $n=1$ 
and $\ell=1$ (as suggested by (\ref{wfi})) and so also decays for large $t$, we conclude, 
using the 1-1 correspondence between 
odd perturbations and the set of $(\q,\f)$'s, together with equations (\ref{F}), (\ref{Q}), (\ref{wfi}) and (\ref{Pfi}), 
that at large $t$
\begin{equation} \label{Flt}
\mathcal{F} \simeq\frac{8Q^2}{r^5} \sqrt{ \frac{4 \pi}{3}}  \sum_{m=\pm1,0} j^{(m)} S_{(\ell=1,m)},
\end{equation}
and
\begin{equation} \label{Qlt}
\mathcal{Q} \simeq \frac{2(Q^2-Mr)(3Mr-2Q^2)}{r^9} \sqrt{ \frac{4 \pi}{3}}  \sum_{m=\pm1,0} j^{(m)} S_{(\ell=1,m)},
\end{equation}
which corresponds to a deformation within the Kerr-Newman (Kerr-Newman de Sitter) family by adding a small amount
of angular momentum. This is consistent with 
the picture that the perturbed black hole settles at large times into a slowly rotating charged black hole.\\

The divergence of $d \mathcal{F}/d\tau$ and $d \mathcal{Q}/d \tau$ 
for free falling radial observers as they approach the Cauchy horizon from region II, 
proved in the previous Section,  supports strong cosmic censorship 
in its purest form, as $\mathcal{Q}$ is a perturbed curvature scalar. This result, however,  has to be taken with caution,  
as the linear perturbation scheme becomes less reliable as linear fields 
 grow. 

\section{Acknowledgements}
This work was partially funded by grants PIP 11220080102479
(Conicet-Argentina) and 30720110101569CB (Universidad Nacional de C\'ordoba). 
J.M.F.T. is supported by a fellowship from Conicet.

\end{document}